\documentclass[letterpaper,11pt,UKenglish]{article}
\usepackage[letterpaper,hmargin=1in,vmargin=1in,footskip=30pt]{geometry}
\usepackage{microtype}%if unwanted, comment out or use option "draft"
\usepackage{amsmath}
\usepackage{amsthm}
\usepackage{tikz}
\usepackage{xcolor}
\usepackage{enumitem}
\usepackage{calc}
\usepackage{latexsym}
\usepackage{hyperref}

\theoremstyle{plain}
\newtheorem{theorem}{Theorem}
\theoremstyle{definition}
\newtheorem{definition}[theorem]{Definition}
\theoremstyle{plain}
\newtheorem{lemma}[theorem]{Lemma}
\newtheorem{obs}[theorem]{Observation}
\newtheorem{claim}[theorem]{Claim}
\newtheorem{corollary}[theorem]{Corollary}

\newcommand{\OPT}{\mathit{OPT}}
\newcommand{\odd}{\mathrm{odd}}
\newcommand{\cost}{\mathrm{cost}}
\newcommand{\weight}{\mathrm{weight}}
\newcommand{\sav}{\mathrm{sav}}

\title{Partitioning Vectors into Quadruples:\\
Worst-Case Analysis of a Matching-Based Algorithm}

\author{Annette M. C. Ficker\thanks{Faculty of Economics and Business, KU Leuven, Leuven, Belgium. E-mail: annette.ficker@kuleuven.be.} \and
Thomas Erlebach\thanks{Department of Informatics, University of Leicester, Leicester, UK. E-mail: t.erlebach@leicester.ac.uk. Supported by a study leave granted by University of Leicester.} \and
Mat\'{u}\v{s} Mihal\'{a}k\thanks{Department of Data Science and Knowledge Engineering, Maastricht University, Maastricht, The Netherlands. E-mail: matus.mihalak@maastrichtuniversity.nl.} \and
Frits C.R. Spieksma\thanks{Department of Mathematics and Computer Science, Eindhoven University of Technology, Eindhoven, The Netherlands. E-mail: f.c.r.spieksma@tue.nl.}}

\begin{document}
\maketitle

\begin{abstract}
Consider a problem where $4k$ given vectors need to be partitioned into $k$ clusters of four vectors each. A cluster of four vectors is called a {\em quad}, and the cost of a quad is the sum of the component-wise maxima of the four vectors in the quad. The problem is to partition the given $4k$ vectors into $k$ quads with minimum total cost. We analyze a straightforward matching-based algorithm, and prove that this algorithm is a $\frac32$-approximation algorithm for this problem. We further analyze the performance of this algorithm on a hierarchy of special cases of the problem, and prove that, in one particular case, the algorithm is a $\frac54$-approximation algorithm. Our analysis is tight in all cases except one.
\end{abstract}

\section{Introduction}

\emph{Partitioning Vectors into Quadruples} (PQ) is the problem of partitioning $4k$ given nonnegative vectors $v_1, \ldots, v_{4k}$, each consisting of $n$ components, into $k$ clusters, each containing exactly four vectors.
We refer to such a cluster of four vectors as a \emph{quadruple} or a \emph{quad} for short. The cost of a quad $Q=\{v_{i_1}, v_{i_2}, v_{i_3}, v_{i_4}\}$ is the sum of the component-wise maxima of the four vectors in the quad. The goal of the problem is to find a partition of the $4k$ vectors into $k$ quads such that the total cost of all quads is minimum.

We will analyze the following matching-based algorithm, called algorithm $A$, that finds a solution to problem PQ by proceeding in two phases. In the first phase, algorithm $A$ builds a complete, edge-weighted graph $G=(V,E)$ that has a node in $V$ for each vector in the instance (hence $|V|=4k$). The weight of an edge equals the sum of the component-wise maxima of the two vectors whose corresponding nodes span the edge. Now, algorithm $A$ computes a minimum-cost perfect matching $M$ in the complete graph $G$, yielding $2k$ vector pairs.
Let $p_1,\ldots,p_{2k}$ be the $2k$ matched vector pairs corresponding
to the computed matching~$M$.

In the second phase, algorithm $A$ builds a complete, edge-weighted graph $G'=(V',E')$ that has a node in $V'$ for each vector pair $p_i$ found in the first phase ($i=1, \ldots, 2k$; $|V'|=2k$). The weight of an edge equals the sum of the component-wise maxima of the two vector pairs whose corresponding nodes span the edge. Now, algorithm $A$ computes a minimum-cost perfect matching $M'$ in the complete graph $G'$. Each of the $k$ edges of $M'$ matches two vector pairs, which naturally induces a quad. The $k$ quads induced by the edges of $M'$ constitute a solution to the problem.
Clearly, $A$ is a polynomial-time algorithm. A rigorous description can be found in Section~\ref{sec:prelim}. It is not hard to see that algorithm $A$ may fail to find an optimum solution for an instance of the problem, i.e., $A$ is not exact, and we are interested in analyzing how far off algorithm $A$'s output can be from an optimum solution.

In this paper we show that $A$ is a $\frac32$-approximation algorithm for problem PQ, and that this bound is tight. We also show that algorithm $A$ has better approximation guarantees for various special cases of problem PQ. In particular, we show that $A$ is a $\frac54$-approximation algorithm for the special case of PQ where each vector is a $\{0,1\}$-vector containing exactly two ones, and the vectors, when seen as expressing the vertex/edge incidence matrix of a graph, correspond to edges of a simple, connected graph. We give a precise overview of our results in Section~\ref{sec:results}.

The remainder of this section introduces some terminology and discusses related work that motivates our research. Section~\ref{sec:prelim} discusses preliminaries and also states our results. The proofs of the upper bounds on the worst-case ratio of algorithm~$A$ for the problem PQ and its special cases can be found in Section~\ref{sec:upperb}, while Section~\ref{sec:lowerb} contains the lower bound results. We conclude in Section~\ref{sec:conclusion}.

\subsection{Terminology and related literature}

Worst-case analysis is a well-established tool to analyze the quality of solutions found by heuristics. We refer to books by Vazirani~\cite{vazirani} and Williamson and Shmoys~\cite{wilshm} for a thorough introduction to the field. We use the following, standard terminology that applies to minimization problems. In the next definition, $A(I)$ stands for the value of the solution to instance $I$ found by algorithm $A$, while $OPT(I)$ stands for the value of an optimum solution to instance $I$.

\begin{definition}
Algorithm $A$ is an $\alpha$-approximation algorithm for a minimization problem P if for every instance $I$ of problem P: (i) algorithm $A$ runs in polynomial-time, and (ii) $A(I) \leq \alpha \cdot  OPT(I)$. We refer to $\alpha$ as an upper bound on the worst-case ratio of algorithm $A$.
\end{definition}

Different problems in various fields are related to problem PQ, and share some of its characteristics. In addition, algorithm $A$ can often be adjusted to work in a particular setting. We now review related literature and provide a number of such examples.

Onn and Schulman~\cite{onnschulman2001} consider a problem where a given set of vectors in $n$-dimensional space needs to be partitioned in a given number of clusters. The number of vectors in a cluster (its {\em size}) is not specified, and in addition, they assume that the objective function, which is to be maximized, is convex in the sum of the vectors in the same cluster. Their framework contains many different problems with diverse applications, and they show, for their setting, strongly-polynomial time, exact algorithms. This is in contrast to our problem which is NP-hard (cf.\ Section~\ref{sec:complexity}).

Another problem, distinct from, yet related to, our problem, comes from computational biology, and is described in Figuero et al.~\cite{figueroaetal2005}. Here, a component of a vector is a 0 or a 1 or an ``N''. In this setting neither the size of a cluster, nor the number of clusters is fixed; the goal is to find a partition of the set of vectors into a minimum number of clusters while satisfying the condition that a pair of vectors that is in the same cluster can only differ at a component where at least one of them has the value N. They prove hardness of this problem, and analyze the approximation behavior of heuristics for this problem.

Hochbaum and Levin \cite{hoclev/2010} describe a problem in the design of optical networks that is related to our special case where each vector is a $\{0,1\}$-vector containing two ones. In essence, their problem is to cover the edges of a given bipartite graph by a minimum number of 4-cycles. They observe that this problem is a special case of unweighted 4-set cover; they give a $(\frac{13}{10}+\epsilon)$-approximation algorithm (using local search), and analyze the performance of a greedy algorithm for a more general version of the problem. Our problem differs from theirs in the sense that we deal with a partitioning problem, where there is a weight for each set; in addition, our problem does not necessarily have a bipartite structure, nor do our quads need to correspond to 4-cycles.

Our problem is also intimately related to a problem occurring in wafer-to-wafer yield optimization (see, e.g., Reda et al.~\cite{redaetal} for a description). Central in this application is the production of so-called {\em waferstacks}, which can be seen as a set of superimposed wafers. In our context, a wafer can be represented by a vector. A wafer consists of many dies, each of which can be in two states: either functioning, i.e., good (which corresponds to a component in the vector with value `0'), or malfunctioning, i.e., bad (which corresponds to a component in the vector with value '1'). The quality of a waferstack is measured by simply counting the number of components that have only 0's in the wafers contained in the waferstack. The goal is to partition the set of wafers into waferstacks (clusters) such that total quality is as high as possible. In this application, however, there are different types of wafers, and a waferstack needs to consist of one wafer of each type. This would correspond to an a priori given partition of the vectors. In addition, a typical waferstack consists in practice of many, i.e., more than 4, wafers. Dokka et al.~\cite{dokkaetal2014} analyze the worst-case behavior of different algorithms that have as a common feature solving assignment problems repeatedly. The case where there are three types of wafers, and the problem is to find waferstacks that are triples containing one wafer of each type is investigated in Dokka et al.~\cite{dokkaetal2013}; for a particular objective function, they describe a $\frac43$-approximation algorithm.

A special, yet very relevant special case of our problem is one where the edges of a given graph need to be partitioned into subsets each containing four edges (see Section~\ref{sec:prelim} for a precise description). Indeed, from a graph-theoretical perspective, there is quite some interest and literature in partitioning the edge-set of a graph, i.e., to find an edge-decomposition. In fact, edge-decompositions where each cluster has prescribed size have already been studied in e.g. J\"unger et al.~\cite{juengeretal1985}. Thomassen~\cite{thomassen2008} studies the existence of edge-decompositions into paths of length 4, and Barat and Gerbner~\cite{barat+gerbner2014} even study edge-decompositions where each cluster is isomorphic to a tree consisting of 4 edges.

\section{Preliminaries}
\label{sec:prelim}

\subsection{About problem PQ: special cases and complexity}
\label{sec:complexity}

We first observe that,
for the analysis of algorithm $A$, we can restrict ourselves to instances of problem PQ where the $4k$ vectors are $\{0,1\}$ vectors. Notice that we call a vector {\em nonnegative} when each of its entries is nonnegative.

\begin{lemma}
\label{lem:reduction01}
Each instance of problem PQ with arbitrary (rational) nonnegative vectors can be
reduced to an instance of problem PQ with $\{0,1\}$ vectors.
\end{lemma}

\begin{proof}
If the components of the input vectors are non-negative integers, we can
reduce the problem to the problem with $\{0,1\}$-vectors as follows: Consider any
component $i$, $1 \leq i \leq n$, and let $M_i$ be the largest value in the
$i$-th component of any input vector. Then replace in each vector component $i$ by $M_i$ components, and a vector with value $x$ in
the $i$-th component is replaced by a vector that has
$x$ ones followed by $M_i-x$ zeros in the $M_i$ components that replace component~$i$. The resulting vectors are
$\{0,1\}$-vectors, and the cost of any set of
original vectors is the same as the cost of the corresponding
set of modified vectors.
If the input vectors have non-negative rational values, we
first multiply all vectors by the lowest common denominator
of all the rational numbers to make all vector components
integers, and then use the reduction described above.
\end{proof}

This shows that for a worst-case analysis of algorithm $A$, it is sufficient to consider $\{0,1\}$-vectors only. Indeed, any worst-case ratio of $A$ shown to hold for $\{0,1\}$-vectors holds, using the argument of Lemma~\ref{lem:reduction01}, for arbitrary rational nonnegative vectors.
However, the reduction described in the proof of Lemma~\ref{lem:reduction01} is not polynomial. We only need the lemma for the purpose of the analysis; and of course, algorithm $A$ can work directly with the original input vectors.

Thus, from hereon we restrict ourselves, without loss of generality, to the case of binary vectors. There are various special cases of PQ that are of independent interest. We will describe the particular special case in brackets following `PQ'; we distinguish the following special cases.
\begin{itemize}
    \item Problem PQ$(\#1 \in \{1,2\})$. The case where each vector contains either one or two 1's; all other components have value 0. It will turn out that, at least in terms of the worst-case behavior of algorithm $A$, this special case displays the same behavior as the general problem PQ.
    \item Problem PQ$(\#1=2)$. The case where each binary vector contains exactly two 1's. Instances of this type can be represented by a multi-graph $F$ with $n$ nodes, each node corresponding to a component of a vector. Each vector is then represented by an edge spanning the two nodes that correspond to components with value 1. Of course, now a quad can be seen as a set of four edges, and its cost equals the number of nodes in the subgraph induced by these four edges.
    \item Problem PQ$(\#1=2, \mbox{distinct})$. The case where the graph $F$ is a simple graph. Equivalently, this means that each vector contains exactly two 1's and the vectors are pairwise distinct.
    \item Problem PQ$(\#1=2, \mbox{distinct, connected})$. We distinguish a further special case by demanding that the graph $F$ is also connected.

\end{itemize}

Clearly, the special cases are ordered, in the sense that each next one is a special case of its predecessor.

Although our interest is on the worst-case behavior of algorithm $A$, it is relevant to establish the computational complexity of problem PQ. We prove that even its special case PQ$(\#1=2, \mbox{distinct, connected})$ is NP-hard. This fact shows that no polynomial-time algorithm for problem PQ can be exact, unless P=NP.

\begin{theorem}
\label{complexity}
PQ$(\#1=2, \mbox{distinct, connected})$ is NP-hard.
\end{theorem}

\begin{proof}
As noted by Hochbaum and Levin~\cite{hoclev/2010}, the following problem was shown to be NP-complete by Holyer~\cite{Holyer/1981}. Given a connected, bipartite graph $G=(V,E)$, where $|E|=4k$ for some $k \in \mathcal{N}$, does there exist a partition of the edge-set $E$ such that each set is isomorphic to a cycle on 4 nodes, i.e., a $C_4$? We refer to this decision problem as EP$C_4$.

Given an instance of EP$C_4$, we build the following instance of PQ$(\#1=2, \mbox{distinct}, \\ \mbox{connected})$.
There are $4k$ vectors, each of length $|V|$. The $|V|$ components of each vector correspond to the nodes in $V$. Each edge in $E$ gives rise to a vector whose entries are 0, except in the two components that correspond to the nodes spanning the edge; these components have value 1.
This specifies all $4k$ vectors.
The question is: does there exist a solution of this instance of PQ$(\#1=2, \mbox{distinct, connected})$ with cost at most $4k$?

We claim that an instance of EP$C_4$ is a yes-instance if and only if there exists a solution to PQ$(\#1=2, \mbox{distinct, connected})$ with cost at most $4k$. Indeed, if the instance of EP$C_4$ is a yes-instance, the four edges of each $C_4$ directly correspond to four vectors making up a quadruple with cost of 4, leading to a total cost of $4k$.

Consider now a solution to PQ$(\#1=2, \mbox{distinct, connected})$, i.e., a set of $k$ quads, with total cost $4k$. Since any four vectors are pairwise distinct, it follows that the four edges corresponding to each quadruple must span at least four vertices, i.e., each quadruple must have cost at least 4. And since the total cost equals $4k$, it follows that each quadruple must have cost exactly 4. Finally, since the only possibility for four edges to span four nodes in a simple bipartite graph is a $C_4$, it follows that a partition into $C_4$'s must exist.
\end{proof}

\subsection{About algorithm A: notation and properties}

Recall that, in our analysis, we may assume that all vectors are $\{0,1\}$-vectors. Let $v_i\vee v_j$ denote the vector that is the component-wise maximum of the two vectors $v_i$ and $v_j$, i.e.:
\[
v_i\vee v_j = (\mbox{max}(v_{i,1}, v_{j,1}), \mbox{max}(v_{i,2}, v_{j,2}), \ldots, \mbox{max}(v_{i,n}, v_{j,n})).
\]
Here, $v_{i,\ell}$ denotes the $\ell$-th component of vector $v_i$ ($\ell=1, \ldots,n$).
We use $|v_i|$ to denote the number of ones in vector $v_i$ ($1 \leq i \leq 4k$), i.e.:
\[
|v_i| = \sum_{\ell=1}^n v_{i,\ell}\,.
\]
The cost of a quad $Q=\{v_1,v_2,v_3,v_4\}$ is then $\cost(Q)=|v_1\vee v_2\vee v_3 \vee v_4|$. For a pair $p=\{v_1,v_2\}$ of vectors, we set $\cost(p)=|v_1\vee v_2|$.

For two vectors $v_i$ and $v_j$, let $\sav(v_i,v_j)$ (the ``savings'' made by combining $v_i$ and $v_j$) denote the number of common ones in $v_i$ and $v_j$, i.e.:
\[
\sav(v_i,v_j) = \sum_{\ell=1}^n \mbox{min}(v_{i,\ell}, v_{j,\ell}).
\]
If $p=\{v_1,v_2\}$ and $p'=\{v_3,v_4\}$ are pairs of vectors, we also write
$\sav(p,p')$ for $\sav(v_1\vee v_2,v_3\vee v_4)$.

The following observation concerning two $\{0,1\}$-vectors $u$ and $v$ is immediate:

\begin{obs}
\label{obs:basic}
$|u|+|v| = \sav(u,v) + |u \vee v|$.
\end{obs}

\begin{proof}
Recall that we may assume that all vectors are $\{0,1\}$-vectors (Lemma~\ref{lem:reduction01}).
Let us partition the set of components that make up the vectors $u$ and $v$ into four sets:
\begin{itemize}
\item Those with a `1' in $u$, and a `0' in $v$: say there are $k_{u,\bar{v}}$ of them.
\item Those with a `0' in $u$, and a `1' in $v$: say there are $k_{\bar{u},v}$ of them.
\item Those with a `1' in $u$, and a `1' in $v$: say there are $k_{u,v}$ of them.
\item Those with a `0' in $u$, and a `0' in $v$: say there are $k_{\bar{u},\bar{v}}$ of them.
\end{itemize}
Obviously, since $|u| = k_{u,\bar{v}} + k_{u,v}$, $|v| = k_{\bar{u},v} + k_{u,v}$, $\sav(u,v) = k_{u,v}$, and $|u \vee v| = k_{u,\bar{v}} +  k_{\bar{u},v} + k_{u,v}$, the claim follows.
\end{proof}

Let us revisit the description of Algorithm $A$. In the first phase, it computes a minimum-cost perfect matching $M$ in the complete graph $G$ on the given $4k$ vectors, where the weight
of the edge between vectors $v_i$ and $v_j$ is set to $|v_i \vee v_j|$.
Let $p_1,\ldots,p_{2k}$ be the $2k$ matched vector pairs corresponding
to the computed matching~$M$, and let $\cost(M)$ denote the cost of the matching~$M$. For $1\le i\le 2k$, let $v_{i}^1$ and $v_{i}^2$ be the two vectors
in the vector pair~$p_i$, and let $v_i'=v_{i}^1\vee v_{i}^2$.

In the second phase, Algorithm $A$ computes a minimum-cost perfect matching $M'$ in the complete graph $G'$ on the $2k$ vector pairs, where the weight of the edge between pairs $p_i$ and $p_j$ is set to $|v_i'\vee v_j'|$.
The quads corresponding to $M'$ are output as a solution.
Let $\cost(M')$ be the cost of matching~$M'$.

\begin{obs}
\label{obs:mmo}
$A(I)=\cost(M')$ and $\cost(M')\le \cost(M)$.
\end{obs}

\begin{lemma}
\label{lem:sav}
In the first phase of algorithm $A$, we can equivalently set the weight of the edge between $v_i$ and $v_j$ to be $-\sav(v_i,v_j)$. Similarly, in the second phase of algorithm $A$, we can set the weight of the edge between $p_i$ and $p_j$ to be $-\sav(v_i',v_j')$.
\end{lemma}

\begin{proof}
For the first phase, it follows from Observation~\ref{obs:basic} that the cost of any perfect matching $M$ can be written as:
\[
\cost(M) = \sum_{(v_i,v_j)\in M} |v_i \vee v_j| = \sum_{i=1}^{4k}|v_i| - \sum_{(v_i,v_j)\in M} \sav(v_i,v_j).
\]
Hence, finding a matching $M$ that minimizes
$\sum_{(v_i,v_j)\in M} |v_i \vee v_j|$ is equivalent to finding
a matching $M$ that minimizes $\sum_{(v_i,v_j)\in M} -\sav(v_i,v_j)$.

In the second phase, the cost of any perfect matching $M'$ is
\begin{eqnarray}
\nonumber
\cost(M') = \sum_{(v_i',v_j')\in M'} |v_i' \vee v_j'| =
\sum_{(v_i',v_j')\in M'} (|v_i'|+|v_j'|-\sav(v_i',v_j'))= \\
\label{eq:costM'}
\sum_i |v_i'| -\sum_{(v_i',v_j')\in M'} \sav(v_i',v_j')  =
\cost(M) - \sum_{(v_i',v_j')\in M'} \sav(v_i',v_j').
\end{eqnarray}

Therefore, finding a matching $M'$ that minimizes
$\sum_{(v_i',v_j')\in M'} |v_i' \vee v_j'|$ is equivalent to finding
a matching $M'$ that minimizes $\sum_{(v_i',v_j')\in M'} -\sav(v_i',v_j')$.
\end{proof}

Let $\mbox{weight}(M')$ denote the total savings of the perfect matching $M'$, i.e.,
\[
\mbox{weight}(M') = \sum_{(v_i',v_j')\in M'} \sav(v_i',v_j').
\]
Then, it follows from Equation~(\ref{eq:costM'}) that
\begin{equation}
\label{eq:sav}
\cost(M')=\cost(M)-\sum_{(v_i',v_j')\in M'} \sav(v_i',v_j') = \cost(M) - \weight(M').
\end{equation}

\noindent
Observation~\ref{obs:mmo} and Equation (\ref{eq:sav}) imply:
\begin{corollary}
\label{cor:sav}
$A(I)=\cost(M)-\weight(M')$.
\end{corollary}

In view of this corollary, it follows that if we can show that $\cost(M)\le B$ and
$\weight(M')\ge S$ for some bounds $B$ and $S$, we can conclude
that $A(I)\le B-S$.

Two vectors $u$ and $v$ are {\em identical} when $u=v$, and a pair of identical vectors is called an {\em identical pair}. In the following we show that among the set of minimum-cost perfect matchings, there is one that contains a maximum number of identical pairs.

\begin{lemma}
\label{lem:identical}
There is a minimum-cost perfect matching in $G$, as well as in $G'$, that contains
a maximum number of identical pairs.
\end{lemma}

\begin{proof}
Assume that, in a minimum-cost perfect matching, there are two identical vectors $u$ and $v$ that are not matched to each other; instead, let $u$ be matched to some vector $a$, and $v$ be matched to some vector $b$. We need to prove that, when $u=v$, the cost of the pairs $\{u,a\}$ and $\{v,b\}$ is at least as large as the cost of the pairs $\{u,v\}$ and $\{a,b\}$.

\begin{eqnarray*}
|u \vee a| + |v \vee b| & = &
\sav((u \vee a),(v \vee b)) + |(u \vee a) \vee (v \vee b)|\\
& \geq & |u| + |u \vee a \vee b|\\
& \geq & |u| + |a \vee b|.
\end{eqnarray*}
The first equality holds by Observation~\ref{obs:basic}, the first inequality follows from $u=v$, and the second inequality is trivial.
\end{proof}

Thus, in the implementation of our algorithm $A$, we can first greedily match pairs of identical vectors as long as they exist, and then use any standard minimum-cost perfect matching algorithm to compute a perfect matching of the remaining vectors.

\subsection{Our results}
\label{sec:results}

In this paper, we show the following bounds on the worst-case ratio of algorithm $A$ (see Table~\ref{table:results} for a summary).

\begin{theorem}
\label{theo:PQ}
Algorithm $A$ is a $\frac32$-approximation algorithm for problem PQ, and this bound is tight.
\end{theorem}

\begin{theorem}
\label{theo:PQatmosttwoones}
Algorithm $A$ is a $\frac32$-approximation algorithm for problem PQ$(\#1\in \{1,2\})$, and this bound is tight.
\end{theorem}

\begin{theorem}
\label{theo:PQexactlytwoones}
Algorithm $A$ is a $\frac43$-approximation algorithm for problem PQ$(\#1=2)$, and this bound is tight.
\end{theorem}

\begin{theorem}
\label{theo:PQexactlytwoonesanddis}
Algorithm $A$ is a $\frac{13}{10}$-approximation algorithm for problem PQ$(\#1=2, \mbox{distinct})$, and its worst-case ratio is at least $\frac{5}{4}$.
\end{theorem}

\begin{theorem}
\label{theo:PQexactlytwoonesanddisandcon}
Algorithm $A$ is a $\frac54$-approximation algorithm for problem PQ$(\#1=2,$ distinct,
connected$)$, and this bound is tight.
\end{theorem}

\begin{table}[ht]
\centering
\begin{tabular}{|l||l|l|}
\hline
Problem name &  Lower Bound & Upper Bound \\
\hline
\  &\  &\  \\[-11pt]
PQ &  $\frac32$ & $\frac32$ (Lemma~\ref{ubPQ}) \\[2pt]
PQ$(\#1 \in \{1,2\})$ &  $\frac32$ (Observation~\ref{wcinstPQatmosttwoones}) & $\frac32$ \\[2pt]
PQ$(\#1=2)$ &  $\frac43$ (Observation~\ref{wcinstPQexactlytwoones}) & $\frac{4}{3}$ (Lemma~\ref{ubPQexactlytwoones})\\[2pt]
PQ$(\#1=2, \mbox{distinct})$ &  $\frac54$ & $\frac{13}{10}$ (Lemma~\ref{ubPQexactlytwoonesanddist})
\\[2pt]
PQ$(\#1=2, \mbox{distinct, connected})$ &  $\frac54$ (Observation~\ref{wcinstPQexactlytwoonesanddistandconn}) & $\frac54$ (Lemma~\ref{ubPQexactlytwoonesanddistandconn}) \\
\hline
\end{tabular}
\caption{Overview of bounds on the worst-case ratio of algorithm $A$}
\label{table:results}
\end{table}

Proofs of the theorems are presented in the following sections: The upper bound proofs (Lemmas~\ref{ubPQ}-\ref{ubPQexactlytwoonesanddistandconn}) are given in Section~\ref{sec:upperb}, and the lower bound results (Observations~\ref{wcinstPQatmosttwoones}-\ref{wcinstPQexactlytwoonesanddistandconn}) in Section~\ref{sec:lowerb}.
As an aside, in Section~\ref{sec:badgreedy} we also give instances that show that the worst-case ratio of a natural greedy algorithm is worse than the worst-case ratio of algorithm $A$, both for problem PQ and for problem PQ($\#1=2, \mbox{distinct, connected}$).

\section{Upper bound proofs}
\label{sec:upperb}%
In this section, we prove the upper bounds for the worst-case ratios of algorithm A for problem PQ and its special cases. In
Section~\ref{sec:ubPQ} we prove the upper bound $\frac32$ for the worst-case ratio of Problem PQ, in Section~\ref{sec:ubPQtwoones} we prove the upper bound $\frac{4}{3}$ for the worst-case ratio of Problem PQ$(\#1=2)$, and in Section~\ref{sec:ubPQtwoonesdist} we prove the upper bound $\frac{13}{10}$ for the worst-case ratio of Problem PQ$(\#1=2, \mbox{distinct})$. Finally, in Section~\ref{sec:ubPQtwoonesdistconn} we prove the upper bound $\frac54$ for the worst-case ratio of Problem PQ$(\#1=2,\mbox{distinct, connected})$.

\subsection{Approximation analysis for PQ}
\label{sec:ubPQ}

\begin{lemma}
\label{ubPQ}
The worst-case ratio of algorithm $A$ for PQ is at most~$\frac32$.
\end{lemma}

\begin{proof}
We use the terminology from Section~\ref{sec:prelim}, where $M$ refers to the minimum-cost perfect matching found by $A$ in the first phase based on the costs $|v_i \vee v_j|$, and $M'$ refers to the maximum-weight perfect matching found in the second phase based on the savings $\sav(v_i',v_j')$. As described in Corollary~\ref{cor:sav}, we can express the cost of the solution found by algorithm $A$ as follows:
\[
A(I) = \mbox{cost}(M) - \mbox{weight}(M').
\]

Consider the quads in an optimum solution. By specifying two vector pairs in each quad from the optimum solution, we obtain a matching $\hat{M}$ that we can compare to the matching $M$ found by $A$. Clearly, by the optimality of the first phase's matching of algorithm $A$, we have:
\begin{equation}
\label{aboutcost}
\cost(M) \leq \cost(\hat{M}), \mbox{ for any possible choice of }\hat{M}.
\end{equation}

Further, we will identify potential matches between vector pairs in $M$ with corresponding savings that algorithm $A$ could make in the second phase. These potential matches are represented as edges in an auxiliary graph $H$ whose vertex set is the set of vector pairs resulting from the algorithm's matching of the first phase. Thus we will construct a graph $H=(V',E_1 \cup E_2)$ where the edge-sets $E_1$ and $E_2$ will be described in detail.
The weight of each edge $e$ in the graph $H$, called $w(e)$, represents the savings that algorithm $A$ would realize in the second phase if it were to match the vector pairs that are the endpoints of~$e$. The graph $H$ can be seen as a ``proxy'' for the graph $G'$ that is used in the second phase of algorithm $A$; it will allow us to argue that a certain amount of savings is guaranteed to exist in an optimum second phase matching.

The graph $H$ will be bipartite and have maximum degree~$2$, implying that every cycle in $H$ must have even length. Furthermore, each edge of $H$ will connect two vertices with the same degree. We use $E_1$ to denote the
edges in $H$ whose both endpoints have degree one, with $S_1$ their total weight, i.e., $S_1 = \sum_{e \in E_1} w(e)$. We use $E_2$ to denote the remaining edges in $H$, with $S_2$ representing their total weight, i.e., $S_2 = \sum_{e \in E_2} w(e)$. Our construction will ensure that $E_2$ is a collection of even-length cycles.

\begin{claim}
\label{propertyH}
Consider an undirected, edge-weighted graph $H$ that is bipartite, and has maximum degree 2.
Further, assume that every edge connects two vertices of the same degree.
Let $S_1$ ($S_2$) be the total weight of edges between nodes with degree 1 (degree 2), and let $M_H$ be a maximum-weight matching in $H$.
Then
\begin{equation*}
  \weight(M_H) \geq S_1 + \frac12S_2.
\end{equation*}
\end{claim}
\begin{proof}
The claim follows because there exists a matching in $H$ with that
weight that can be obtained by
taking all edges from $E_1$ and partitioning $E_2$ into two matchings and
taking the one with maximum weight.
\end{proof}

We claim that the matching $M'$ that algorithm $A$ finds in $G'$ in the
second phase has total savings at least $S_1+\frac12 S_2$.
This follows from Claim~\ref{propertyH} because $H$ is a subgraph of $G'$
and we can obtain a perfect matching of $G'$ by taking a maximum-weight
matching of $H$ and matching any remaining vector pairs arbitrarily. Thus, we get:
\begin{equation}
\label{propertyHweight}
  \weight(M') \geq S_1 + \frac12S_2.
\end{equation}

Inequalities (\ref{aboutcost}) and (\ref{propertyHweight}) imply:
\begin{equation}
\label{basicUB}
A(I) = \mbox{cost}(M) - \mbox{weight}(M') \leq \mbox{cost}(\hat{M}) - (S_1+\frac12 S_2).
\end{equation}

Consider the quantity $\cost(\hat{M})-(S_1+\frac12 S_2)$, which - according to (\ref{basicUB}) - serves as an upper bound for the cost of the solution found by algorithm $A$. Informally speaking, we are going to distribute this quantity over the quads from the optimum solution: for each quad $Q$ in the optimum solution, we will define its corresponding ``share'' of $\cost(\hat{M})-(S_1+\frac12 S_2)$ by $\phi_Q$; we will refer to $\phi_Q$ as the {\em contribution reserved for} $Q$. This contribution $\phi_Q$ consists of terms reflecting the contribution to $\hat{M}$, and terms reflecting the contribution to the total savings $S_1+\frac12S_2$. We will show that the choice of $\phi_Q$ satisfies, for each $Q$ from the optimum solution:

\begin{equation}
\label{phiq3over2}
\phi_Q \leq \frac32 \mbox{cost}(Q).
\end{equation}

This leads to the following:
\begin{eqnarray*}
A(I) \leq \mbox{cost}(\hat{M}) - (S_1+\frac12 S_2) = \sum_Q \phi_Q \leq \sum_Q \frac32 \mbox{cost}(Q) = \frac32 OPT(I).
\end{eqnarray*}

The remainder of the proof is devoted to proving the above relationship.
Thus, correctness hinges upon proving that
\begin{enumerate}[labelindent=0pt,labelwidth=\widthof{(iii)},label=\arabic*.,itemindent=1em,leftmargin=!]
\item[(i)] the graph $H$ that we will construct is bipartite, and has maximum degree 2,
\end{enumerate}
and that our choice of $\phi_Q$ satisfies
\begin{enumerate}[labelindent=0pt,labelwidth=\widthof{(iii)},label=\arabic*.,itemindent=1em,leftmargin=!]
\item[(ii)] $\phi_Q \leq \frac32 \mbox{cost}(Q)$ for each $Q$ from the optimum solution, and
\item[(iii)] $\sum_Q \phi_Q = \mbox{cost}(\hat{M}) - (S_1+\frac12 S_2)$.
\end{enumerate}
We now prove (i), (ii), and (iii).

\paragraph*{Proving that the graph $H$ is bipartite, and has maximum degree 2}

\begin{definition}
A \textbf{lucky pair} is a pair of vectors that are in the same quad in the optimum solution
and that are matched by algorithm $A$ in the first phase.
\end{definition}

Let us first construct the edge-set $E_1$ of the graph $H$. Consider an optimal quad $Q$ that contains two lucky pairs $p_1$ and $p_2$. By definition, $p_1$ and $p_2$ correspond to nodes in $H$, and we add the edge $e=(p_1,p_2)$ to $E_1$, and we set its weight equal to the corresponding savings:  $w(e)=\sav(p_1,p_2)$. Notice that neither node $p_1$ nor $p_2$ will be incident to any other edge in $H$.

Let us now proceed with the edge-set $E_2$ of the graph $H$.
The edge set $E_2$ of $H$ will be constructed as follows.
Consider an auxiliary multi-graph $K$ whose vertices are the quads
of the optimal solution that contain at most one lucky pair.
For every vector pair $(v_1,v_2)$ that is matched by the algorithm in
the first phase, add an edge
$(Q_1,Q_2)$ to $K$, where $Q_i$ is the optimal quad that contains
the vector $v_i$ for $i=1,2$. We say that this edge $(Q_1,Q_2)$
\emph{corresponds} to the pair $(v_1,v_2)$. If $(v_1,v_2)$ is
a lucky pair, the edge added to $K$ is a self-loop at the node corresponding to the quad that contains $v_1$ and $v_2$. Every
edge of $K$ corresponds to a vertex of the auxiliary graph $H$,
as the vertices of $H$ are the pairs of vectors matched by the algorithm
in the first phase.
Note that each vertex in $K$ has degree four, where a self-loop
contributes $2$ to the degree of the vertex to which it is attached.
As every vertex of $K$ has even degree,
every connected component of $K$ admits an Eulerian cycle.
Note that each Eulerian cycle of a connected component of $K$
has an even number of edges as each vertex in the component
has degree~$4$ and the number of edges in a multi-graph with
self-loops is equal to half the sum of the vertex degrees.

Pick an arbitrary Eulerian cycle (possibly including self-loops) in each component of~$K$. We will use
these Eulerian cycles to determine edges to be added to $E_2$
in $H$ in
such a way that $H$ is bipartite and has maximum degree~$2$.
Orient each Eulerian cycle in an arbitrary way into a directed
cycle. For every pair of consecutive edges $(Q_i,Q_{i+1})$ and
$(Q_{i+1},Q_{i+2})$ on such a cycle,
where $(Q_i,Q_{i+1})$ corresponds to $p_1=\{v_i,v_{i+1}\}$
and $(Q_{i+1},Q_{i+2})$ to $p_2=\{v_{i+1}',v_{i+2}\}$, add
the edge $e=(p_1,p_2)$ to~$H$, and set its weight to $w(e) = \sav(p_1,p_2)$. The edge $(p_1,p_2)$ is considered
to be added for the optimal quad $Q_{i+1}$, and its savings
are used only in the analysis of that optimal quad.
See Figure~\ref{fig:QKH} for an example of graph $K$ and $H$.

\begin{figure}[ht]
\centering
\scalebox{0.8}{
\begin{tikzpicture}

\def\xShi{0};
\def\yShi{0};
\def\inter{0.3};
\def\size{0.75};
\node[label = east:{OPT quads and first phase matching $M$:}] at (\xShi-2,\yShi+2) {};
\node[label = above:{$Q_1$}] at (\xShi,\yShi+\size) {};
\draw (\xShi,\yShi) circle (\size cm);
\node(q11) at (\xShi -\inter,\yShi+\inter) [circle,fill=black,inner sep=0pt,minimum size=5pt] {};
\node(q12) at (\xShi +\inter,\yShi+\inter) [circle,fill=black,inner sep=0pt,minimum size=5pt] {};
\node(q13) at (\xShi -\inter,\yShi-\inter) [circle,fill=black,inner sep=0pt,minimum size=5pt] {};
\node(q14) at (\xShi +\inter,\yShi-\inter) [circle,fill=black,inner sep=0pt,minimum size=5pt] {};
\draw[thick](q11)--(q13);
\node[label=west:{$p_0$}] at  (\xShi-.75,\yShi) {};
\def\xShi{2};
\def\yShi{0};
\node[label = above:{$Q_2$}] at (\xShi,\yShi+\size) {};
\node[label=above:{$p_1$}] at  (\xShi-1,\yShi+0.5) {};
\node[label=below:{$p_2$}] at  (\xShi-1,\yShi-0.5) {};
\draw (\xShi,\yShi) circle (\size cm);
\node(q21) at (\xShi -\inter,\yShi+\inter) [circle,fill=black,inner sep=0pt,minimum size=5pt] {};
\node(q22) at (\xShi +\inter,\yShi+\inter) [circle,fill=black,inner sep=0pt,minimum size=5pt] {};
\node(q23) at (\xShi -\inter,\yShi-\inter) [circle,fill=black,inner sep=0pt,minimum size=5pt] {};
\node(q24) at (\xShi +\inter,\yShi-\inter) [circle,fill=black,inner sep=0pt,minimum size=5pt] {};
\def\xShi{4};
\def\yShi{0};
\node[label = above:{$Q_3$}] at (\xShi,\yShi+\size) {};
\node[label=above:{$p_3$}] at  (\xShi-1,\yShi+0.5) {};
\node[label=below:{$p_4$}] at  (\xShi-1,\yShi-0.5) {};
\draw (\xShi,\yShi) circle (\size cm);
\node(q31) at (\xShi -\inter,\yShi+\inter) [circle,fill=black,inner sep=0pt,minimum size=5pt] {};
\node(q32) at (\xShi +\inter,\yShi+\inter) [circle,fill=black,inner sep=0pt,minimum size=5pt] {};
\node(q33) at (\xShi -\inter,\yShi-\inter) [circle,fill=black,inner sep=0pt,minimum size=5pt] {};
\node(q34) at (\xShi +\inter,\yShi-\inter) [circle,fill=black,inner sep=0pt,minimum size=5pt] {};
\draw[thick](q32)--(q34);
\node(q11a)[label=east:{$p_5$}] at  (\xShi+.75,\yShi) {};

\draw[thick](q12)--(q21);
\draw[thick](q14)--(q23);
\draw[thick](q22)--(q31);
\draw[thick](q24)--(q33);

\def\xShi{0};
\def\yShi{-3.5};
\node[label = east:{Graph $K$:}] at (\xShi-2,\yShi+1.5) {};
\node[label=west:{$p_0$}] at  (\xShi-.75,\yShi) {};
\node(q1)[label=above:{$Q_1$}] at  (\xShi,\yShi) [circle,fill=black] {};
\node(q12a)[label=above:{$p_1$}] at  (\xShi+1,\yShi+0.5) {};
\node(q12b)[label=below:{$p_2$}] at  (\xShi+1,\yShi-0.5) {};
\node(q2)[label=above:{$Q_2$}] at  (2,\yShi) [circle,fill=black] {};
\node(q23a)[label=above:{$p_3$}] at  (\xShi+3,\yShi+0.5) {};
\node(q23b)[label=below:{$p_4$}] at  (\xShi+3,\yShi-0.5) {};
\node(q3)[label=above:{$Q_3$}] at  (4,\yShi) [circle,fill=black] {};
\node[label=east:{$p_5$}] at  (\xShi+4.75,\yShi) {};
\draw ((\xShi-.5,\yShi) circle(0.5*\size);
\draw (q1) .. controls (q12a) .. (q2);
\draw (q1) .. controls (q12b) .. (q2);
\draw (q2) .. controls (q23a) .. (q3);
\draw (q2) .. controls (q23b) .. (q3);
\draw ((\xShi+4.5,\yShi) circle(0.5*\size);

\def\xShi{0};
\def\yShi{-7};
\node[label = east:{Two possible graphs $H$:}] at (\xShi-2,\yShi+1.5) {};
\node(p0)[label=west:{$p_0$}] at  (\xShi-.75,\yShi) [circle,fill=black,inner sep=0pt,minimum size=7.5pt] {};
\node(p1)[label=above:{$p_1$}] at  (\xShi+1,\yShi+.5) [circle,fill=black,inner sep=0pt,minimum size=7.5pt] {};
\node(p2)[label=below:{$p_2$}] at  (\xShi+1,\yShi-.5) [circle,fill=black,inner sep=0pt,minimum size=7.5pt] {};
\node(p3)[label=above:{$p_3$}] at  (\xShi+3,\yShi+.5) [circle,fill=black,inner sep=0pt,minimum size=7.5pt] {};
\node(p4)[label=below:{$p_4$}] at  (\xShi+3,\yShi-.5) [circle,fill=black,inner sep=0pt,minimum size=7.5pt] {};
\node(p5)[label=east:{$p_5$}] at  (\xShi+4.75,\yShi) [circle,fill=black,inner sep=0pt,minimum size=7.5pt]{};
\draw(p0)--(p1);
\draw(p3)--(p1);
\draw(p0)--(p2);
\draw(p4)--(p2);
\draw(p4)--(p5);
\draw(p3)--(p5);

\def\xShi{0};
\def\yShi{-10};
\node(p0)[label=west:{$p_0$}] at  (\xShi-.75,\yShi) [circle,fill=black,inner sep=0pt,minimum size=7.5pt] {};
\node(p1)[label=above:{$p_1$}] at  (\xShi+1,\yShi+.5) [circle,fill=black,inner sep=0pt,minimum size=7.5pt] {};
\node(p2)[label=below:{$p_2$}] at  (\xShi+1,\yShi-.5) [circle,fill=black,inner sep=0pt,minimum size=7.5pt] {};
\node(p3)[label=above:{$p_3$}] at  (\xShi+3,\yShi+.5) [circle,fill=black,inner sep=0pt,minimum size=7.5pt] {};
\node(p4)[label=below:{$p_4$}] at  (\xShi+3,\yShi-.5) [circle,fill=black,inner sep=0pt,minimum size=7.5pt] {};
\node(p5)[label=east:{$p_5$}] at  (\xShi+4.75,\yShi) [circle,fill=black,inner sep=0pt,minimum size=7.5pt]{};
\draw(p0)--(p1);
\draw(p4)--(p1);
\draw(p0)--(p2);
\draw(p2)--(p3);
\draw(p4)--(p5);
\draw(p3)--(p5);

\end{tikzpicture}
}
\caption{Example of graphs \texorpdfstring{$K$}{K} and \texorpdfstring{$H$}{H} with \texorpdfstring{$p_0$}{p0} and \texorpdfstring{$p_5$}{p5} as lucky pairs}
\label{fig:QKH}
\end{figure}

\begin{claim}
\label{cla:Hbipartite}
The graph $H$ is bipartite, and has maximum degree 2.
\end{claim}
\begin{proof}
Observe that, for every component of $K$, the edges added to $H$
form a single cycle consisting of an even number of edges. This is true because
the vertices of $H$ correspond to edges of $K$ and the edges of $H$ connect
consecutive edges of an Eulerian cycle of $K$, which has even length.
\end{proof}

Clearly, when building $H$ as described above, different choices can be made, since there might be different Eulerian cycles possible in $K$ in the case of a quad not containing a lucky pair. We need to be explicit about these different possibilities. Indeed, consider any quad $Q=\{v_1,v_2,v_3,v_4\}$ that does not contain a lucky pair. Let $p_i$ be the vector pair containing $v_i$ that was matched by the algorithm in the first phase, for $1\le i\le 4$. Let $(p_1,p_3)$
and $(p_2,p_4)$ be the edges added to $H$, and observe that they lie on a single even-length cycle $C$ in $H$ (namely, the cycle created in $H$ from the Eulerian cycle of the component of $K$ that contains $Q$). Assume, without loss of generality, that after removing the edges $(p_1,p_3)$ and $(p_2,p_4)$ from $H$, the cycle $C$ splits into two paths in $H$, one between node $p_1$ and node $p_2$, and one between node $p_3$ and node $p_4$. Define $\{\{p_1,p_2\},\{p_3,p_4\}\}$ to be the \emph{good partition} associated with $Q$.

Observe now that replacing the edges $(p_1,p_3)$  and $(p_2,p_4)$ in $H$ by the edges $(p_1,p_4)$ and $(p_2,p_3)$ in $H$ replaces the cycle $C$
by another cycle $C'$ consisting of the same number of edges. This means that
any one of the two possible combinations of two independent edges
between a vertex on one side of the good partition and a vertex
on the other side of the good partition can be chosen
for inclusion in~$E_2$, while maintaining the property that
$E_2$ consists of even-length cycles. For example, in Figure~\ref{fig:QKH}, we show the two possible cycles of even length for graph $H$.

Observe that the discussion following Claim~\ref{cla:Hbipartite} has identified a collection of graphs $H$, each satisfying the conditions of Claim~\ref{cla:Hbipartite}.

\paragraph*{Proving that $\phi_Q \leq \frac32 \cost(Q) \mbox{ for each quad }Q$}

Now we are ready to analyse the contribution to $\phi$ of each quad~$Q$ from the optimum solution. There are three types of quads in an optimum solution:
\begin{itemize}
\item those quads that contain two lucky pairs; let us refer to this set of quads as $O_2$,
\item those quads that contain one lucky pair; let us refer to this set of quads as $O_1$,
\item those quads that contain no lucky pairs; let us refer to this set of quads as $O_0$.
\end{itemize}

Let us first consider the quads from the set $O_2$. Let $Q \in O_2$ equal $\{p_1,p_2\}$. We set
\begin{equation}
\label{eq:phiQO0}
\phi_Q = \mbox{cost}(p_1) + \mbox{cost}(p_2) - \sav(p_1,p_2).
\end{equation}

Since $\cost(Q) = \mbox{cost}(p_1) + \mbox{cost}(p_2) - \sav(p_1,p_2) = \phi_Q$, it trivially follows that:
\begin{equation}
\label{twoluckybound}
\phi_Q \leq \frac32 \cost(Q) \mbox{ for each quad }Q \in O_2.
\end{equation}

Next, we consider the quads containing a single lucky pair, i.e., the quads from $O_1$. Thus, with $Q=\{u_1,u_2,v_1,v_2\}$, let the lucky pair be
$p=\{u_1,u_2\}$, and let $p_1=\{v_1,w_1\}$ and $p_2=\{v_2,w_2\}$ be the other nodes in $H$ that contain the vectors from quad $Q$. Clearly, the edge-set $E_2$ contains the edge $(p,p_1)$ as well as $(p,p_2)$, with weights respectively $\sav(p,p_1)$ and $\sav(p,p_2)$. Observe that there are no other edges in $H$ incident to node $p$.

For each $Q=\{u_1,u_2,v_1,v_2\} \in O_1$, we set:
\begin{equation}
\label{eq:phiQO1}
\phi_Q = \cost(p)+|v_1\vee v_2|-\frac12(\sav(p,p_1)+\sav(p,p_2)).
\end{equation}

The cost of $Q$ is:
\begin{equation}
\label{eq:costQO1}
\cost(Q)=\cost(p)+|v_1\vee v_2|-\sav(p,v_1\vee v_2).
\end{equation}
Notice that
$\cost(Q)\ge\cost(p)$ and
$\cost(Q)\ge|v_1\vee v_2|$,
and therefore:
\begin{equation}
\label{eq:costQO1UB}
\cost(Q)\ge\frac12(\cost(p)+|v_1\vee v_2|).
\end{equation}

Further, we have:
\begin{equation}
\label{eq:savO1}
\sav(p,v_1\vee v_2) \le
\sav(p,v_1)+\sav(p,v_2) \le
\sav(p,p_1)+\sav(p,p_2).
\end{equation}
Combining (\ref{eq:phiQO1}), (\ref{eq:costQO1UB}), and (\ref{eq:savO1}) gives, for each quad $Q \in O_1$:
\begin{eqnarray}
\label{oneluckybound}
\nonumber
\phi_Q &=& \cost(p)+|v_1\vee v_2|-\frac12(\sav(p,p_1)+\sav(p,p_2))\\
\nonumber
& \le &
\cost(p)+|v_1\vee v_2|-\frac12\sav(p,v_1\vee v_2) \\
\nonumber
& = &
\frac12(\cost(p)+|v_1\vee v_2|)
+\frac12(\cost(p)+|v_1\vee v_2|-\sav(p,v_1\vee v_2)) \\
& \le & \cost(Q)
+\frac12\cdot\cost(Q)
= \frac32\cdot\cost(Q).
\end{eqnarray}

Now, consider a quad $Q=\{v_1,v_2,v_3,v_4\}$ from $O_0$, i.e., a quad with no lucky pairs. As mentioned before, the term $\phi_Q$ consists of terms reflecting the contribution to $\hat{M}$, and terms reflecting the contribution to the savings. As there are three ways to choose two vector pairs from $Q$, the contribution to the matching $\hat{M}$ can be realized by any of the
three expressions
$|v_1\vee v_2|+|v_3\vee v_4|$,
$|v_1\vee v_3|+|v_2\vee v_4|$,
$|v_1\vee v_4|+|v_2\vee v_3|$.

Let $p_i$ be the vector pair matched by the algorithm in the first phase that contains $v_i$, for $1\le i\le 4$. Assume that $\{\{p_1,p_2\},\{p_3,p_4\}\}$ is the good partition associated with~$Q$; it follows that we can choose for inclusion in $H$ either the edges $(p_1,p_3)$ and $(p_2,p_4)$ or the edges $(p_1,p_4)$ and $(p_2,p_3)$. Depending on which pair of edges we choose for inclusion in $H$, the term in $\phi_Q$ that reflects the contribution to the savings equals either $\frac12(\sav(p_1,p_3)+\sav(p_2,p_4))$ or $\frac12(\sav(p_1,p_4)+\sav(p_2,p_3))$. Observe that we have:

\begin{eqnarray}
\label{boundsonsavings1}
\sav(p_1,p_3)+\sav(p_2,p_4) \geq \sav(v_1,v_3) + \sav(v_2,v_4) \\
\label{boundsonsavings2}
\sav(p_1,p_4)+\sav(p_2,p_3) \geq \sav(v_1,v_4) + \sav(v_2,v_3).
\end{eqnarray}

Summarizing, depending on which pairs of vectors from $Q$ we
put in $\hat{M}$ (there are three possibilities), and which two edges we add to $H$ (there are two possibilities), we can get (using (\ref{boundsonsavings1})-(\ref{boundsonsavings2}))
any of the following six bounds on the contribution $\phi_Q$ of $Q$:

\begin{eqnarray}
\phi_Q & \leq & |v_1\vee v_2|+|v_3\vee v_4| -\frac12(\sav(v_1,v_3)+\sav(v_2,v_4)), \label{eq:b1}\\
\phi_Q & \leq & |v_1\vee v_2|+|v_3\vee v_4| -\frac12(\sav(v_1,v_4)+\sav(v_2,v_3)), \label{eq:b2}\\
\phi_Q & \leq & |v_1\vee v_3|+|v_2\vee v_4| -\frac12(\sav(v_1,v_3)+\sav(v_2,v_4)), \label{eq:b3}\\
\phi_Q & \leq & |v_1\vee v_3|+|v_2\vee v_4| -\frac12(\sav(v_1,v_4)+\sav(v_2,v_3)), \label{eq:b4}\\
\phi_Q & \leq & |v_1\vee v_4|+|v_2\vee v_3| -\frac12(\sav(v_1,v_3)+\sav(v_2,v_4)), \label{eq:b5}\\
\phi_Q & \leq & |v_1\vee v_4|+|v_2\vee v_3| -\frac12(\sav(v_1,v_4)+\sav(v_2,v_3)). \label{eq:b6}
\end{eqnarray}
We can choose the pairs of vectors for $\hat{M}$ and the two edges
we add to $H$ for $Q$ in such a way that the smallest of these six
bounds becomes an upper bound on the contribution reserved for $Q$.
For that choice, all right-hand sides (\ref{eq:b1})-(\ref{eq:b6}) are upper bounds on the contribution $\phi_Q$.

Since we may assume that each vector $v_i$ is a $\{0,1\}$-vector, and since the right-hand sides of (\ref{eq:b1})-(\ref{eq:b6}) involve the four vectors of quad $Q$, there are $2^4=16$ possible configurations for the 4 values of a particular component in the four vectors. Thus, for each $j=0,1,\ldots, 15$, we can write its binary expansion as $j=8 b_{j,1} + 4 b_{j,2}+ 2 b_{j,3} + b_{j,4}$. We denote by $n_j$ the number of components of the vectors of the quad whose values equal the binary expansion of $j$, i.e., the number of components $r$ with $v_{1,r} = b_{j,1}$, $v_{2,r} = b_{j,2}$, $v_{3,r} = b_{j,3}$, $v_{4,r} = b_{j,4}$.

We can now express each of the relevant quantities as linear functions
of the $n_j$:
\begin{eqnarray}
\cost(Q) & = & n_1+n_2+\cdots+n_{15}, \label{eq:Q}\\
|v_1\vee v_2| & = & n_4+n_5+n_6+n_7+n_8+n_9+n_{10}\nonumber\\&&\mbox{}+n_{11}+n_{12}+n_{13}+n_{14}+n_{15},\\
|v_1\vee v_3| & = & n_2+n_3+n_6+n_7+n_8+n_9+n_{10}\nonumber\\&&\mbox{}+n_{11}+n_{12}+n_{13}+n_{14}+n_{15},\\
|v_1\vee v_4| & = & n_1+n_3+n_5+n_7+n_8+n_9+n_{10}\nonumber\\&&\mbox{}+n_{11}+n_{12}+n_{13}+n_{14}+n_{15},\\
|v_2\vee v_3| & = & n_2+n_3+n_4+n_5+n_6+n_7+n_{10}\nonumber\\&&\mbox{}+n_{11}+n_{12}+n_{13}+n_{14}+n_{15},\\
|v_2\vee v_4| & = & n_1+n_3+n_4+n_5+n_6+n_7+n_9\nonumber\\&&\mbox{}+n_{11}+n_{12}+n_{13}+n_{14}+n_{15},\\
|v_3\vee v_4| & = & n_1+n_2+n_3+n_5+n_6+n_7+n_9\nonumber\\&&\mbox{}+n_{10}+n_{11}+n_{13}+n_{14}+n_{15},\\
\sav(v_1,v_3) & = & n_{10}+n_{11}+n_{14}+n_{15},\\
\sav(v_2,v_4) & = & n_5+n_7+n_{13}+n_{15},\\
\sav(v_1,v_4) & = & n_9+n_{11}+n_{13}+n_{15},\\
\sav(v_2,v_3) & = & n_6+n_7+n_{14}+n_{15}. \label{eq:s23}
\end{eqnarray}

Using the identities in (\ref{eq:Q})--(\ref{eq:s23}), we can write the inequalities (\ref{eq:b1}), (\ref{eq:b4}), and (\ref{eq:b6}) as
follows:
\begin{eqnarray}
\phi_Q & \le & n_1+n_2+n_3+n_4+\frac32n_5+2n_6+\frac32n_7+n_8\nonumber\\&&\mbox{}+2n_9+\frac32n_{10}+\frac32n_{11}+n_{12}+\frac32n_{13}+\frac32n_{14}+n_{15}, \label{eq:zb1}\\
\phi_Q & \le & n_1+n_2+2n_3+n_4+n_5+\frac32n_6+\frac32n_7+n_8\nonumber\\&&\mbox{}+\frac32n_9+n_{10}+\frac32n_{11}+2n_{12}+\frac32n_{13}+\frac32n_{14}+n_{15}, \label{eq:zb4}\\
\phi_Q & \le & n_1+n_2+2n_3+n_4+2n_5+\frac12n_6+\frac32n_7+n_8\nonumber\\&&\mbox{}+\frac12n_9+2n_{10}+\frac32n_{11}+2n_{12}+\frac32n_{13}+\frac32n_{14}+n_{15}. \label{eq:zb6}
\end{eqnarray}

Multiplying (\ref{eq:zb1}) by $\frac12$, (\ref{eq:zb4}) by $\frac14$, and
(\ref{eq:zb6}) by $\frac14$, and then adding the three inequalities, we get, for each quad $Q \in O_0$:
\begin{eqnarray}
\label{noluckybound}
\nonumber
\phi_Q & \le & n_1+n_2+\frac32n_3+n_4+\frac32n_5+\frac32n_6+\frac32n_7+n_8\\
\nonumber
&&\mbox{}+\frac32n_9+\frac32n_{10}+\frac32n_{11}+\frac32n_{12}+\frac32n_{13}+\frac32n_{14}+n_{15}\\
\nonumber
& =& \frac32 \cdot \sum_{j=1}^{15} n_j - \frac12n_1 - \frac12n_2 - \frac12n_4 - \frac12n_8 - \frac12n_{15}\\
\nonumber
& =& \frac32 \cdot \cost(Q) - \frac12n_1 - \frac12n_2 - \frac12n_4 - \frac12n_8 - \frac12n_{15}\\
& \le& \frac32 \cdot \cost(Q).
\end{eqnarray}

Thus, (\ref{twoluckybound}), (\ref{oneluckybound}), and (\ref{noluckybound}) show that indeed:
\[
\phi_Q \leq \frac32 \mbox{cost}(Q) \mbox{ for each } Q \mbox{  from the optimum solution.}
\]

\paragraph*{Proving that $\sum_Q \phi_Q = \mbox{cost}(\hat{M}) - (S_1+\frac12 S_2)$}

Since a quad contains either zero, one, or two lucky pairs, it follows that the expressions for $\phi_Q$ given in (\ref{eq:phiQO0}), (\ref{eq:phiQO1}), (\ref{eq:b1})-(\ref{eq:b6}) contain the terms that jointly sum up to $\cost(\hat{M})$. Additionally, it is not difficult to verify that the construction of the graph $H$ is such that the savings on the edges of $H$ (as defined in (\ref{eq:phiQO0}), (\ref{eq:phiQO1}), (\ref{eq:b1})-(\ref{eq:b6})) sum up to $S_1+\frac12S_2$.

The proof is complete.
\end{proof}

\subsection{Approximation analysis for \texorpdfstring{PQ$(\#1=2)$}{PQ(\#1=2)}}
\label{sec:ubPQtwoones}

\begin{lemma}
\label{ubPQexactlytwoones}
The worst-case ratio of algorithm $A$ for PQ$(\#1=2)$ is at most~$\frac{4}{3}$.
\end{lemma}

\begin{proof}
Recall that an instance of PQ$(\#1=2)$ is a multi-graph $G$ with $4k$ edges. We assume that the algorithm chooses in the first phase a matching that matches as many pairs of duplicates as possible (see Lemma~\ref{lem:identical}).

Our analysis follows the structure of the proof of Lemma~\ref{ubPQ}, but we obtain a better bound on the ratio between $\phi_Q$ and $\mbox{cost}(Q)$ by exploiting the restricted set of configurations that are possible for a quad $Q$ in the optimal solution if every vector in the given instance of the problem has exactly two ones.

By specifying two edge pairs in each quad from the optimal
solution, we obtain again a matching $\hat{M}$ whose cost is an upper bound
on the cost of the matching $M$ computed by the algorithm in the
first phase, i.e. (cf.\ (\ref{aboutcost})):

\begin{equation}
\label{aboutcostno2}
\cost(M) \leq \cost(\hat{M}), \mbox{ for any possible choice of }\hat{M}.
\end{equation}

Furthermore, we again construct an auxiliary graph $H=(V',E_1\cup E_2)$ that represents potential matches between vector pairs in $M$ with corresponding savings that algorithm~$A$ could make in the second phase. The weight $w(e)$ of each edge $e$ in the graph $H$ represents the savings that algorithm $A$ would realize in the second phase if it were to match the vector pairs that are the endpoints of~$e$. The graph $H$ is again bipartite (we will use the concept of good partitions introduced in the proof of Lemma~\ref{ubPQ} to ensure this) and has maximum degree~$2$, and the edges in $E_1$ connect vertices of degree~$1$ and the edges in $E_2$ connect vertices of degree~$2$. Letting $S_1$ denote the total weight of $E_1$, and $S_2$ the total weight of $E_2$, it follows from Claim~\ref{propertyH} that the maximum-weight matching $M_H$ in $H$ satisfies $\mbox{weight}(M_H)\ge S_1+\frac12 S_2$, and hence the weight of the matching $M'$ that algorithm~$A$ finds in the second phase satisfies:

\begin{equation}
\label{propertyHweight2}
  \weight(M') \geq S_1 + \frac12S_2.
\end{equation}
Thus, we again have the relationship
\begin{eqnarray*}
A(I) \leq \mbox{cost}(\hat{M}) - (S_1+\frac12 S_2)
\end{eqnarray*}
and distribute the value on the right-hand side over the quads of the optimal solution, with each quad $Q$
receiving a share $\phi_Q$ that is referred to as the contribution reserved for~$Q$.
We will show that $\phi_Q\le \frac{4}{3}\mbox{cost}(Q)$ holds for all quads $Q$ of the
optimal solution, implying that:
\begin{eqnarray*}
A(I) \leq \mbox{cost}(\hat{M}) - (S_1+\frac12 S_2) = \sum_Q \phi_Q \le \sum_Q \frac{4}{3} \mbox{cost}(Q) = \frac{4}{3}\OPT(I)
\end{eqnarray*}

Thus, correctness hinges upon proving that
\begin{enumerate}[labelindent=0pt,labelwidth=\widthof{(iii)},label=\arabic*.,itemindent=1em,leftmargin=!]
\item[(i)] $\phi_Q \leq \frac{4}{3} \mbox{cost}(Q)$ for each $Q$ from the optimum solution,
\item[(ii)] $\sum_Q \phi_Q = \mbox{cost}(\hat{M})-(S_1+\frac12 S_2)$, and
\item[(iii)] the graph $H$ that we construct is bipartite and has maximum degree 2.
\end{enumerate}

\paragraph*{Proving that $\phi_Q \leq \frac{4}{3} \mbox{cost}(Q)$ for each $Q$ from the optimum solution}
Note that a quad contains four edges of multi-graph $G$ and its cost equals the number of vertices in the subgraph induced by these four edges. Thus, in this proof we write quad $Q$ from the optimum solution as a set of four edges of $G$, i.e., $Q=(e_1,e_2,e_3,e_4)$.
We distinguish the following cases for the quads depending on the value of $\cost(Q)$:

\begin{itemize}

\begin{figure}[ht]
\centering
\begin{tikzpicture}
\tikzstyle{mynode}=[circle,draw, fill = white]
\node (n11) at (0,4.75){};
\node (n12) at (0,4.9){};
\node (n13) at (0,5.1){};
\node (n14) at (0,5.25){};
\node (n21) at (2,4.75){};
\node (n22) at (2,4.9){};
\node (n23) at (2,5.1){};
\node (n24) at (2,5.25){};
\draw (n11) -- (n21);
\draw (n12) -- (n22);
\draw (n13) -- (n23);
\draw (n14) -- (n24);
\node at (0,5) [mynode] {1};
\node  at (2,5) [mynode] {2};
\end{tikzpicture}
\caption{$\cost(Q)=2$; 1 and 2 are nodes in $G$}
\label{fig:PQMQ2}
\end{figure}
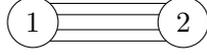

\item $\cost(Q)=2$. This means all four edges are identical, see Figure~\ref{fig:PQMQ2}.
As we may assume that the algorithm matches as many duplicate vectors (edges of $G$) as possible (Lemma \ref{lem:identical}),
we can say that the algorithm matches $p_1=\{e_1,e_2\}$ and $p_2=\{e_3,e_4\}$ in the first phase.

We select both $p_1$ and $p_2$ to be part of matching $\hat{M}$, with total cost $2+2=4$.
If the algorithm, in the second phase, matches $p_1$ and $p_2$, the algorithm will make a saving of $2$. Thus, we add edge $(p_1,p_2)$ with savings $2$ to edge-set $E_1$ in $H$, as both $p_1$ and $p_2$ do not appear in any other quads.
Hence,
$$ \phi_Q = \cost(p_1) + \cost(p_2) - \sav(p_1,p_2) = 2 < \frac{4}{3} \cost(Q).$$

\item $\cost(Q)=3$. The edges in $G$ belonging to $Q$ can have one of the following three structures, see Figure~\ref{fig:PQMQ3}:

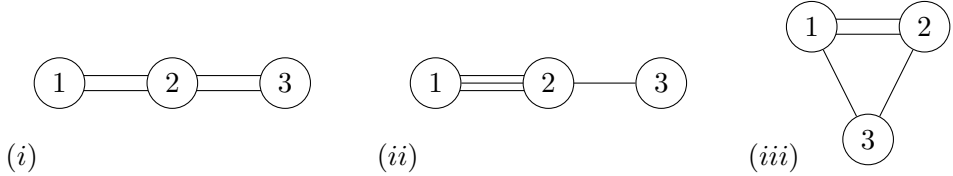
\begin{figure}[ht]
\centering
\begin{tikzpicture}
\tikzstyle{mynode}=[circle,draw, fill = white]

% center height = 0
\node  at (-1.5,-1)  {$(i)$};
\node (n11) at (-1,-0.1){};
\node (n12) at (-1,0.1){};
\node (n31) at (2,-0.1){};
\node (n32) at (2,0.1){};
\draw (n11) -- (n31);
\draw (n12) -- (n32);
\node(n1) at (-1,0) [mynode] {1};
\node(n2)  at (0.5,0) [mynode] {2};
\node(n3) at (2,0) [mynode] {3};

\node  at (3.5,-1)  {$(ii)$};
\node (n11) at (4,-0.1){};
\node (n12) at (4,0.1){};
\node (n21) at (5.5,-0.1){};
\node (n22) at (5.5,0.1){};
\draw (n11) -- (n21);
\draw (n12) -- (n22);
\node(n3) at (7,0) [mynode] {3};
\node(n1) at (4,0) [mynode] {1};
\draw (n3)--(n1);
\node(n2)  at (5.5,0) [mynode] {2};

\node  at (8.5,-1)  {$(iii)$};
\node (n13) at (9,0.85){};
\node (n14) at (9,0.65){};
\node (n23) at (10.5,0.85){};
\node (n24) at (10.5,0.65){};
\draw (n13) -- (n23);
\draw (n14) -- (n24);
\node(n1) at (9,0.75) [mynode] {1};
\node(n2)  at (10.5,0.75) [mynode] {2};
\node(n3) at (9.75,-0.75) [mynode] {3};
\draw (n1) -- (n3);
\draw (n2) -- (n3);

\end{tikzpicture}
\caption{Quads with $\cost(Q)=3$; 1, 2 and 3 are nodes in $G$}
\label{fig:PQMQ3}
\end{figure}

\begin{enumerate}[labelindent=0pt,labelwidth=\widthof{(iii)},label=\arabic*.,itemindent=1em,leftmargin=!]
\item[$(i)$] If quad $Q$ contains $e_1 =(1,2)$, $e_2 = (1,2)$, $e_3=(2,3)$ and $e_4=(2,3)$, we select both $p_1=\{(1,2),(1,2)\}$ and $p_2=\{(2,3),(2,3)\}$ to be part of matching $\hat{M}$, with a total cost of $2+2=4$. Furthermore, we add the edge $(p_1,p_2)$ with savings $1$ to the edge-set $E_1$ of $H$. Hence,
$$\phi_Q \leq \cost(p_1) + \cost(p_2) -\sav(p_1,p_2) = 3 < \frac{4}{3} \cost(Q).$$

\item[$(ii)$] If quad $Q$ contains $e_1 =(1,2)$, $e_2 = (1,2)$, $e_3=(1,2)$ and $e_4=(2,3)$, we may assume that the algorithm
has matched $p_1=\{(1,2),(1,2)\}$ in the first phase. Furthermore, we select $p_1$ and $p_2 = \{e_3,e_4\}$ to be part of matching $\hat{M}$, with total cost $2+3=5$.

In a worst case scenario, $p_2$ is not contained in $M$, hence the algorithm has matched $e_3$ to another vector (edge in $G$), say $x_3$, and $e_4$ to say $x_4$.
We define $p_3 = \{ e_3,x_3\}$ and $p_4 = \{e_4,x_4\}$, which are nodes in the auxiliary graph $H$.
We add edges $(p_1,p_3)$ and $(p_1,p_4)$ to edge set $E_2$, with total savings at least $2+1=3$.
Hence,

$$ \phi_Q = \cost(p_1) + \cost(p_2) - \frac12\Big(\sav(p_1,p_3)+\sav(p_1,p_4)\Big) \le 5-1.5 < \frac{4}{3} \cost(Q).$$

\item[$(iii)$] If quad $Q$ contains $e_1 =(1,2)$, $e_2 = (1,2)$, $e_3=(1,3)$ and $e_4=(2,3)$, we use a similar argument as above.
We may assume that the algorithm has matched $p_1=\{(1,2),(1,2)\}$ in the first phase and select $p_1$ and $p_2 = \{e_3,e_4\}$ to be part of matching $\hat{M}$, with total cost $2+3=5$.

In a worst case scenario, $p_2$ is not contained in $M$, hence the algorithm has matched $e_3$ to say $x$, and $e_4$ to say $y$.
We define $p_3 = \{ e_3,x\}$ and $p_4 = \{e_4,y\}$, which are nodes in auxiliary graph $H$.
We add edges $(p_1,p_3)$ and $(p_1,p_4)$ to edge set $E_2$, with total savings at least $1+1=2$.
Hence,

$$ \phi_Q = \cost(p_1) + \cost(p_2) - \frac12\Big(\sav(p_1,p_3)+\sav(p_1,p_4)\Big) \le 4 \le \frac{4}{3} \cost(Q).$$

\end{enumerate}
\item $\cost(Q)=4$. The edges in $G$ belonging to $Q$ can have one of the following six structures, see Figure~\ref{fig:PQMQ4}:

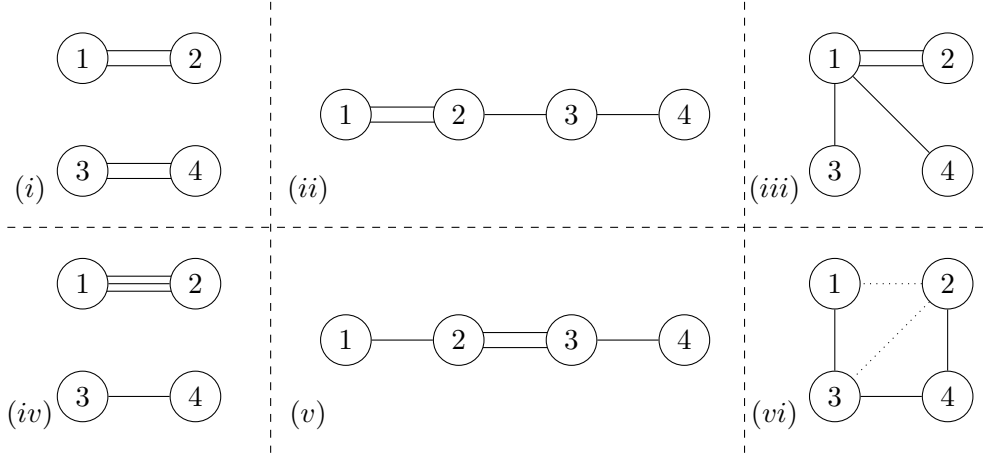
\begin{figure}[ht]
\centering
\begin{tikzpicture}
\tikzstyle{mynode}=[circle,draw, fill = white]

% center height = 0\node  at (-1.5,-1)  {$(i)$};
\node  at (-1.7,-1)  {$(i)$};
\node (n11) at (-1,0.85){};
\node (n12) at (-1,0.65){};
\node (n21) at (0.5,0.85){};
\node (n22) at (0.5,0.65){};
\node (n31) at (-1,-0.85){};
\node (n32) at (-1,-0.65){};
\node (n41) at (0.5,-0.85){};
\node (n42) at (0.5,-0.65){};
\draw (n11)--(n21);
\draw (n12)--(n22);
\draw (n31)--(n41);
\draw (n32)--(n42);
\node(n1) at (-1,0.75) [mynode] {1};
\node(n2)  at (0.5,0.75) [mynode] {2};
\node(n3) at (-1,-0.75) [mynode] {3};
\node(n4) at (0.5,-0.75) [mynode] {4};

\node  at (2,-1)  {$(ii)$};
\node (n11) at (2.5,-0.1){};
\node (n12) at (2.5,0.1){};
\node (n21) at (4,-0.1){};
\node (n22) at (4,0.1){};
\draw (n11) -- (n21);
\draw (n12) -- (n22);
\node(n4) at (7,0) [mynode] {4};
\node(n2)  at (4,0) [mynode] {2};
\node(n1) at (2.5,0) [mynode] {1};
\draw (n2)--(n4);
\node(n3)  at (5.5,0) [mynode] {3};

\node  at (8.2,-1)  {$(iii)$};
\node (n41) at (9,0.85){};
\node (n42) at (9,0.65){};
\node (n31) at (10.5,0.85){};
\node (n32) at (10.5,0.65){};
\draw (n41) -- (n31);
\draw (n42) -- (n32);
\node(n1) at (9,-0.75) [mynode] {3};
\node(n2)  at (10.5,-0.75) [mynode] {4};
\node(n3) at (9,0.75) [mynode] {1};
\node(n4) at (10.5,0.75) [mynode] {2};
\draw (n1)--(n3);
\draw (n2)--(n3);

\draw[dashed](-2,-1.5) --(11,-1.5);
% center height = -3

\node  at (-1.7,-4)  {$(iv)$};
\node (n11) at (-1,-2.15){};
\node (n12) at (-1,-2.35){};
\node (n21) at (0.5,-2.15){};
\node (n22) at (0.5,-2.35){};
\draw (n11)--(n21);
\draw (n12)--(n22);
\node(n1) at (-1,-2.25) [mynode] {1};
\node(n2)  at (0.5,-2.25) [mynode] {2};
\node(n3) at (-1,-3.75) [mynode] {3};
\node(n4) at (0.5,-3.75) [mynode] {4};
\draw (n1)--(n2);
\draw (n3)--(n4);

\node  at (2,-4)  {$(v)$};
\node (n11) at (5.5,-3.1){};
\node (n12) at (5.5,-2.9){};
\node (n21) at (4,-3.1){};
\node (n22) at (4,-2.9){};
\draw (n11) -- (n21);
\draw (n12) -- (n22);
\node(n4) at (7,-3) [mynode] {4};
\node(n1) at (2.5,-3) [mynode] {1};
\node(n2)  at (4,-3) [mynode] {2};
\node(n3)  at (5.5,-3) [mynode] {3};
\draw (n1)--(n2);
\draw (n3)--(n4);

\node  at (8.2,-4)  {$(vi)$};
\node(n1) at (9,-2.25) [mynode] {1};
\node(n2)  at (10.5,-2.25) [mynode] {2};
\node(n3) at (9,-3.75) [mynode] {3};
\node(n4) at (10.5,-3.75) [mynode] {4};
\draw (n1)--(n3);
\draw[dotted] (n2)--(n1);
\draw[dotted] (n2)--(n3);
\draw (n4)--(n3);
\draw (n4)--(n2);

\draw[dashed](1.5,1.5) --(1.5,-4.5);
\draw[dashed](7.8,1.5) --(7.8,-4.5);

\end{tikzpicture}
\caption{Quads with $\cost(Q)=4$; 1, 2, 3 and 4 are nodes in $G$}
\label{fig:PQMQ4}
\end{figure}

\begin{enumerate}[labelindent=0pt,labelwidth=\widthof{(ii)\&(iii)},label=\arabic*.,itemindent=1em,leftmargin=!]
\item[$(i)$] If quad $Q$ contains $e_1 =(1,2)$, $e_2 = (1,2)$, $e_3=(3,4)$ and $e_4=(3,4)$, we select both $p_1=\{(1,2),(1,2)\}$ and $p_2=\{(3,4),(3,4)\}$ to be part of matching $\hat{M}$, with a total cost of $2+2=4$. We add the edge $(p_1,p_2)$ with savings $0$ to the edge-set $E_1$ of $H$.
Hence,
$$\phi_Q \leq \cost(p_1) + \cost(p_2) = 4 < \frac{4}{3} \cost(Q).$$

\item[$(ii) \& (iii)$] If the graph induced by quad $Q$ contains two identical edges, $e_1$ and $e_2$, and two adjacent edges, $e_3$ and $e_4$, then select $p_1=\{e_1, e_2\}$ and $p_2=\{e_3,e_4\}$ to be part of matching $\hat{M}$, with a total cost of $2+3=5$.
We can assume that algorithm $A$ has matched $e_1$ and $e_2$ in the first phase. In the worst case, the algorithm has matched
$e_3$ to another edge $x_3$ and $e_4$ to another edge~$x_4$. Let $p_3=(e_3,x_3)$ and $p_4=(e_4,x_4)$, and add the edges $(p_1,p_3)$
and $(p_1,p_4)$ with total savings at least $1+0=1$ to the edge-set $E_2$ of $H$.
Hence,
$$\phi_Q \leq \cost(p_1) + \cost(p_2) - \frac12 = 4.5 < \frac{4}{3} \cost(Q).$$

\item[$(iv) \&  (v)$] If the graph induced by quad $Q$ contains two identical edges, $e_1$ and $e_2$, and two non-adjacent edges, $e_3$ and $e_4$, then we select $p_1=\{e_1, e_2\}$ and $p_2=\{e_3,e_4\}$ to be part of matching $\hat{M}$, with a total cost of $2+4=6$.

In a worst case scenario, $p_2$ is not contained in $M$, hence the algorithm has matched $e_3$ to another vector (edge in $G$), say $x_3$, and $e_4$ to say $x_4$. We define that $p_3 = \{ e_3,x_3\}$ and $p_4 = \{e_4,x_4\}$, which are nodes in auxiliary graph $H$.
We add edges $(p_1,p_3)$ and $(p_1,p_4)$ to edge set $E_2$ in $H$. Either $e_3$ is a duplicate of $e_1$ (case $(iv)$), in which case the total saving is at least $2+0=2$, or both edges are adjacent to both $e_1$ and $e_2$ (case $(v)$), in which case the total saving is at least $1+1=2$.
Hence,
$$\phi_Q = \cost(p_1) + \cost(p_2) - \frac12\Big(\sav(p_1,p_3)+\sav(p_1,p_4)\Big) \le 5 < \frac{4}{3}  \cost(Q).$$

\item[$(vi)$] Suppose quad $Q$ contains $e_1 =(1,2)$, $e_2 = (1,3)$, $e_3=(2,4)$ and $e_4=(3,4)$ (if $Q$ contains $e_1 = (2,3)$, $e_2 = (1,3)$, $e_3=(2,4)$ and $e_4=(3,4)$ we use the exact same argument).
We select $p_1 = \{e_1,e_2\}$ and $p_2 = \{e_3,e_4\}$ to be part of matching $\hat{M}$, with total cost $3+3=6$.

In a worst case scenario, no edge pair of $G$ in $Q$ is contained in $M$, hence we say the algorithm has matched $e_i$ to $x_i$ and define $p_{e_{i}}=\{ e_i,x_i\}$ for $i \in \{1, \dots,4\}$.  Note that each of the $p_{e_{i}}$ is a node in auxiliary graph $H$.
If the good partition associated with Q (as defined in the proof of Lemma~\ref{ubPQ}) is $\{\{p_{e_{1}}, p_{e_{2}}\},\{p_{e_{3}}, p_{e_{4}}\}\}$,
we add edges $(p_{e_{1}}, p_{e_{3}})$ and $(p_{e_{2}}, p_{e_{4}})$ to edge-set $E_2$ of $H$, with total savings at least $1+1=2$.
Otherwise, we add edges $(p_{e_{1}}, p_{e_{2}})$ and $(p_{e_{3}}, p_{e_{4}})$ to edge set $E_2$, with total savings at least $1+1=2$.
Hence,
$$ \phi_Q \le \cost(p_1) + \cost(p_2) - 1 = 5 < \frac{4}{3}  \cost(Q).$$

\end{enumerate}

\item $\cost(Q)=5$. If $Q$ contains two identical edges, say $e_1$ and $e_2$, we select $p_1=\{e_1,e_2\}$ and $p_2=\{e_3,e_4\}$ to be part of $\hat{M}$, with a total cost of at most $2+4=6$. We can assume that the algorithm has matched $(e_1,e_2)$ in $M$. In the worst case, the algorithm has not matched $e_3$ to $e_4$ in $M$. Let $p_{e_3}=\{e_3,x_3\}$
and $p_{e_4}=\{e_4,x_4\}$ be the matched pairs in $M$ that contain $e_3$ and $e_4$, respectively. We add the edges $(p_1,p_{e_3})$ and $(p_1,p_{e_4})$ to the edge-set $E_2$ of $H$. Even without taking the savings of those edges into account, we have
$$ \phi_Q \le \cost(p_1) + \cost(p_2) = 6 < \frac43 \cost(Q).$$
Assume now that $Q$ does not contain any two identical edges. Then the graph
induced by $Q=(e_1,e_2,e_3,e_4)$ must fall into one of the four following cases, see Figure~\ref{fig:PQMQ5}:

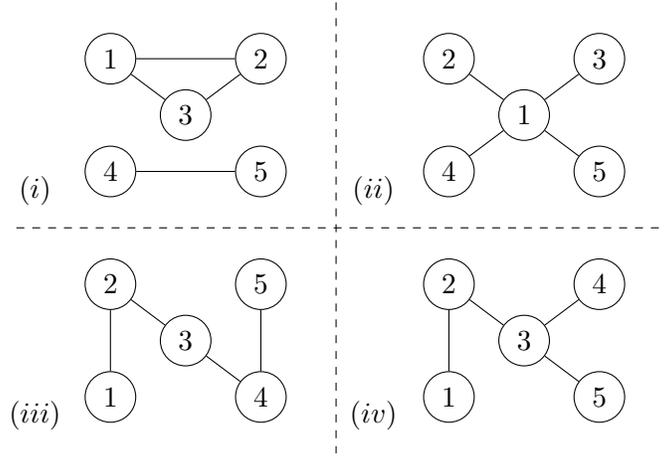
\begin{figure}[ht]
\centering
\begin{tikzpicture}
\tikzstyle{mynode}=[circle,draw, fill = white]

% center height = 0
\node  at (-2,-1)  {$(i)$};
\node(n1) at (-1,0.75) [mynode] {1};
\node(n2)  at (1,0.75) [mynode] {2};
\node(n3) at (0,0) [mynode] {3};
\node(n4) at (-1,-0.75) [mynode] {4};
\node(n5) at (1,-0.75) [mynode] {5};
\draw(n1)--(n2);
\draw(n2)--(n3);
\draw(n1)--(n3);
\draw(n4)--(n5);

\node  at (2.5,-1)  {$(ii)$};
\node(n1) at (3.5,0.75) [mynode] {2};
\node(n2)  at (5.5,0.75) [mynode] {3};
\node(n3) at (4.5,0) [mynode] {1};
\node(n4) at (3.5,-0.75) [mynode] {4};
\node(n5) at (5.5,-0.75) [mynode] {5};
\draw(n1)--(n3);
\draw(n2)--(n3);
\draw(n4)--(n3);
\draw(n5)--(n3);

\draw[dashed](-2.25,-1.5) --(6.5,-1.5);
% center height = -3

\node  at (-2,-4)  {$(iii)$};
\node(n1) at (-1,-2.25) [mynode] {2};
\node(n2)  at (1,-2.25) [mynode] {5};
\node(n3) at (0,-3) [mynode] {3};
\node(n4) at (-1,-3.75) [mynode] {1};
\node(n5) at (1,-3.75) [mynode] {4};
\draw(n1)--(n4);
\draw(n1)--(n3);
\draw(n3)--(n5);
\draw(n2)--(n5);

\node  at (2.5,-4)  {$(iv)$};
\node(n1) at (3.5,-2.25) [mynode] {2};
\node(n2)  at (5.5,-2.25) [mynode] {4};
\node(n3) at (4.5,-3) [mynode] {3};
\node(n4) at (3.5,-3.75) [mynode] {1};
\node(n5) at (5.5,-3.75) [mynode] {5};
\draw(n1)--(n4);
\draw(n1)--(n3);
\draw(n3)--(n2);
\draw(n3)--(n5);

\draw[dashed](2,1.5) --(2,-4.5);

\end{tikzpicture}
\caption{Quads with $\cost(Q)=5$; 1, 2, 3, 4  and 5 are nodes in $G$}
\label{fig:PQMQ5}
\end{figure}

\begin{enumerate}[labelindent=0pt,labelwidth=\widthof{(iii)},label=\arabic*.,itemindent=1em,leftmargin=!]
\item[(i)] $Q$ is a triangle of edges $e_1$, $e_2$, $e_3$ plus a disjoint edge $e_4$. We select $p_1=\{e_1,e_2\}$ and $p_2=\{e_3,e_4\}$ to be part of $\hat{M}$, with a total cost of $3+4=7$. In the worst case, $Q$ does not contain a lucky pair. For $1\le i\le 4$, let $p_{e_i}=\{e_i,x_i\}$ be the pair of matched edges in $M$ that includes~$e_i$. If the good partition associated with $Q$ is $\{\{p_{e_1},p_{e_2}\},\{p_{e_3},p_{e_4}\}\}$, we add the edges $(p_{e_1},p_{e_3})$ and $(p_{e_2},p_{e_4})$ to $E_2$ with savings at least $1+0=1$. Otherwise, add the edges $(p_{e_1},p_{e_2})$ and $(p_{e_3},p_{e_4})$ to $E_2$, again with savings at least $1+0=1$.
We have
$$ \phi_Q \le \cost(p_1) + \cost(p_2) - \frac12 = 6.5 < \frac{4}{3}  \cost(Q).$$
\item[(ii)] $Q$ is a star (i.e., all four edges of $Q$ are incident with the same same vertex).
We select $p_1=\{e_1,e_2\}$ and $p_2=\{e_3,e_4\}$ to be part of $\hat{M}$, with a total cost of $3+3=6$.
In the worst case, $Q$ does not contain a lucky pair. For $1\le i\le 4$, let $p_{e_i}=\{e_i,x_i\}$ be the pair of matched edges in $M$ that includes~$e_i$.
If the good partition associated with $Q$ is $\{\{p_{e_1},p_{e_2}\},\{p_{e_3},p_{e_4}\}\}$, add the edges $(p_{e_1},p_{e_3})$ and $(p_{e_2},p_{e_4})$ to $E_2$, with savings at least $1+1=2$. Otherwise, add the edges $(p_{e_1},p_{e_2})$ and $(p_{e_3},p_{e_4})$ to $E_2$, again with savings at least $1+1=2$.
We have
$$ \phi_Q \le \cost(p_1) + \cost(p_2) - 1 = 5 < \frac{4}{3}  \cost(Q).$$
\item[(iii)] $Q$ is a path of four edges. Assume that the edges appear on the path
in the order $(e_1,e_2,e_3,e_4)$.
We select $p_1=\{e_1,e_2\}$ and $p_2=\{e_3,e_4\}$ to be part of $\hat{M}$, with a total cost of $3+3=6$.
In the worst case, $Q$ does not contain a lucky pair. For $1\le i\le 4$, let $p_{e_i}=\{e_i,x_i\}$ be the pair of matched edges in $M$ that includes~$e_i$.
If the good partition associated with $Q$ is $\{\{p_{e_1},p_{e_2}\},\{p_{e_3},p_{e_4}\}\}$, add the edges $(p_{e_1},p_{e_4})$ and $(p_{e_2},p_{e_3})$ to $E_2$, with savings at least $0+1=1$. Otherwise, add the edges $(p_{e_1},p_{e_2})$ and $(p_{e_3},p_{e_4})$ to $E_2$, with savings at least $1+1=2$.
We have
$$ \phi_Q \le \cost(p_1) + \cost(p_2) - \frac12 = 5.5 < \frac{4}{3}  \cost(Q).$$
\item[(iv)] $Q$ is a tree of diameter~$3$. Assume that $e_1=(1,2)$, $e_2=(2,3)$, $e_3=(3,4)$, $e_4=(3,5)$.
We select $p_1=\{e_1,e_2\}$ and $p_2=\{e_3,e_4\}$ to be part of $\hat{M}$, with a total cost of $3+3=6$.
In the worst case, $Q$ does not contain a lucky pair. For $1\le i\le 4$, let $p_{e_i}=\{e_i,x_i\}$ be the pair of matched edges in $M$ that includes~$e_i$.
If the good partition associated with $Q$ is $\{\{p_{e_1},p_{e_2}\},\{p_{e_3},p_{e_4}\}\}$, add the edges $(p_{e_1},p_{e_4})$ and $(p_{e_2},p_{e_3})$ to $E_2$, with savings at least $0+1=1$. Otherwise, add the edges $(p_{e_1},p_{e_2})$ and $(p_{e_3},p_{e_4})$ to $E_2$, with savings at least $1+1=2$.
We have
$$ \phi_Q \le \cost(p_1) + \cost(p_2) - \frac12 = 5.5 < \frac{4}{3}  \cost(Q).$$
\end{enumerate}

\item $\cost(Q)\ge6$. If we partition $Q$ into two edge pairs $p_1$ and $p_2$ for matching $\hat{M}$, their total cost will be at most $4+4=8$. Even if the savings of the edges that we add to the edge-set $E_2$ of $H$ (while taking into account the good partition associated with $Q$ to ensure that $H$ is bipartite, of course) are zero, we have
$$\phi_Q \leq \cost(p_1) + \cost(p_2) = 8 \leq \frac{4}{3}  \cost(Q).$$
\end{itemize}

\bigskip

\paragraph*{Proving that $H$ is a simple, bipartite graph with maximum degree 2}

An edge has been added to the edge-set $E_1$ of $H$ when considering a quad $Q$ only if $Q$ is a quad with two lucky pairs, and hence those pairs indeed become vertices of degree~$1$ in $H$. If $p=(e_1,e_2)$ is a lucky pair that is in a quad $Q$ with two other edges $e_3,e_4$ that do not form a lucky pair, it becomes the endpoint of exactly two edges in $E_2$: The edges $(p,p')$ and $(p,p'')$, where $p'$ is the pair of edges containing $e_3$ that is matched in $M$, and $p''$ is the pair of edges containing $e_4$ that is matched in $M$. Finally, every edge pair $p=(e_1,e_2)$ that has been matched by algorithm $A$ in $M$ and is not a lucky pair becomes the endpoint of exactly two edges in $E_2$: An edge $(p,p')$ added for quad $Q_1$ and an edge $(p,p'')$ added for quad $Q_2$, where $Q_1$ is the optimal quad containing $e_1$ and $Q_2$ the optimal quad containing~$e_2$. This shows that $H$ is a graph with maximum degree $2$ in which each edge connects two vertices of the same degree. Furthermore, by ensuring for each quad $Q$ that the edges added to $E_2$ are compatible with the good partition associated with $Q$ if $Q$ does not have a lucky pair, we have ensured that $H$ is simple and bipartite (recall that the cycles in $H$ correspond to Eulerian cycles of the connected components of $K$, which have an even number of edges, as discussed in the proof of Lemma~\ref{ubPQ}.

\bigskip

\paragraph*{Proving that $\sum_Q \phi_Q = \cost(\hat{M})-(S_1+\frac12 S_2)$}

For each quad $Q$ we have selected two pairs of edges in $Q$ for inclusion in $\hat{M}$, so the costs of those edge pairs clearly add up to $\cost(\hat{M})$. Furthermore, each edge of $H$ was added to $H$ by some quad $Q$. Furthermore, no two quads could have added the same edge to $H$ (as $H$ is a simple graph), so the sum of the savings of the added edges is indeed $S_1+\frac12 S_2$.

The proof is complete.
\end{proof}

\subsection{Approximation analysis for \texorpdfstring{PQ$(\#1=2, \mbox{distinct})$}{PQ(\#1=2,distinct)}}
\label{sec:ubPQtwoonesdist}

\begin{lemma}
\label{ubPQexactlytwoonesanddist}
The worst-case ratio of algorithm $A$ for PQ$(\#1=2, \mbox{distinct})$ is at most $\frac{13}{10}$.
\end{lemma}

\begin{proof}
Recall that an instance of PQ$(\#1=2, \mbox{distinct})$ is nothing else but a simple graph $F$ with $4k$ edges. Note that the cost of every optimal quad is at least $4$ since $4$ edges in a simple graph touch at least $4$ different vertices. Hence we can repeat the arguments in the proof for Lemma~\ref{ubPQexactlytwoones} for quads of cost 4 and higher.

Then, we get the following results:
\begin{itemize}
\item For $\cost(Q) = 4$, we have $\phi_Q \leq 5 = \frac{5}{4} \cost(Q)< \frac{13}{10}\cdot \cost(Q)$,
\item For $\cost(Q) = 5$, we have $\phi_Q \leq 7-\frac12 = \frac{6.5}{5} \cost(Q)=\frac{13}{10}\cdot \cost(Q)$ ,
\item For $\cost(Q) = 6$, note that the average degree of the subgraph of $G$ induced by $Q$
is $\frac86> 1$, so there must exist a vertex of degree at least~$2$. This means that $Q$
contains two adjacent edges $e_1$ and $e_2$. Denote the remaining edges by $e_3$ and $e_4$.
We can select the pairs $p_1=\{e_1,e_2\}$ and $p_2=\{e_3,e_4\}$ to be part of $\hat{M}$,
giving $\phi_Q \le \cost(p_1)+\cost(p_2)\le 3+4 = 7 = \frac{7}{6} \cost(Q) < \frac{13}{10}\cdot \cost(Q)$.
\item For $\cost(Q) \geq 7$, we have $\phi_Q \leq 8 \le \frac{8}{7}\cost(Q) < \frac{13}{10}\cdot \cost(Q)$.
\end{itemize}
\end{proof}

\subsection{Approximation analysis for \texorpdfstring{PQ(\boldmath$\#1=2, \mbox{distinct, connected}$)}{PQ(\#1=2,distinct,connected)}}
\label{sec:ubPQtwoonesdistconn}

\begin{lemma}
\label{ubPQexactlytwoonesanddistandconn}
Algorithm $A$ is a $\frac54$-approximation algorithm for PQ$(\#1=2, \mbox{distinct, connected})$.

\end{lemma}

\begin{proof}
Recall that an instance of PQ$(\#1=2, \mbox{distinct, connected})$ can be viewed as a simple, connected graph $F$ with $4k$ edges, and that the cost of a quad is the number of vertices spanned by the edges in the quad. Note that the cost of every optimal quad is at least $4$ since $4$ edges in a simple graph touch at least $4$ different vertices. Hence, $\OPT\ge 4k$.
Furthermore, if we can show that there are $z$ quads in the optimal solution that have cost at least $5$,
we get that $\OPT\ge 4(k-z)+5z=4k+z$.

\begin{obs}
$\cost(M) = 6k$.
\end{obs}
\begin{proof}
	The line graph of a connected graph with an even number of edges admits a perfect matching (J\"{u}nger et al.~\cite{juengeretal1985}, Dong et al.~\cite{DYZ/13}).
	Thus, the minimum-cost perfect matching $M$ pairs adjacent edges of the graph. Hence, every pair in $M$ has cost $3$, and thus the cost of $M$ is $2k\cdot 3=6k$.
\end{proof}

Let $p_1,\ldots,p_{2k}$ be the pairs corresponding to~$M$.
Consider the auxiliary graph $H$ with vertex set $V'=\{p_1,\ldots,p_{2k}\}$ in which an edge is added between $p_i$ and $p_j$ if $p_i$ and $p_j$ have
at least one common vertex (implying that matching
$p_i$ to $p_j$ in the matching $M'$ that $A$ computes
in the second phase would create a saving of at least one).
Note that $H$ is connected as $F$ is connected.
Let $\mu$ be the size of a maximum matching in $H$, $1\le \mu\le k$.
Note that the maximum matching of $H$ can be extended to a perfect
matching of $V'$ that makes savings at least $\mu$.
Therefore, we have
$$ A(I) \le 6k - \mu.$$

If $H$ contains a perfect matching, we have $\mu=k$ and hence $A(I)\le 5k$, implying that $A(I)/\OPT(I)\le 5k/(4k)=\frac54$. It remains
to consider the case $\mu<k$.

If a maximum matching in $H$ has size $\mu<k$, the number
of unmatched vertices is $2k-2\mu$. We will show that the optimal
solution then contains at least $k-\mu$ quads with cost
at least~$5$, and hence we have
$\OPT(I) \ge 4k+(k-\mu)=5k-\mu$.
Therefore,
$$
\frac{A(I)}{\OPT(I)}
\le \frac{6k-\mu}{5k-\mu}
\le \frac{5}{4},
$$
where the last inequality follows because $(6k-\mu)/(5k-\mu)$
is maximized if $\mu$ takes its maximum possible value, $\mu=k$.

It remains to show that the optimal solution contains at
least $k-\mu$ quads with cost at least~$5$. Recall that
a maximum matching in $H$ leaves $2k-2\mu$ vertices
unmatched. By the Tutte-Berge formula~\cite{Berge/58}, the number of unmatched
vertices of a maximum matching in $H$ is equal to
$$
\max_{X\subseteq V'} (\odd(H-X) - |X|),
$$
where $\odd(H-X)$ is the number of connected components
of $H-X$ that have an odd number of vertices ($H-X$ is the graph that results when the nodes in $X$, and their incident edges, are removed from $H$). Hence, there exists a set $X\subseteq V'$ such that
$\odd(H-X) - |X| = 2k-2\mu$. Let $d=\odd(H-X)$,
and let $O_1,O_2,\ldots,O_d$ denote the $d$ odd components
of $H-X$.
We have
$$
2k-2\mu = d-|X|\text{.}
$$

For a subgraph $S$ of $H$, let $E_F(S)$ denote the
set of edges of $F$ that are contained in the
edge pairs that form the vertex set of $S$ (recall
that the vertices of $H$ are pairs of edges from $F$).
Note that $|E_F(O_i)|\bmod 4 = 2$
for $1\le i\le d$ as $O_i$ contains an odd number
of edge pairs. Therefore, each $E_F(O_i)$ contains
at least two edges that are contained in optimal
quads that do not only contain edges from $E_F(O_i)$.
If such a quad contains three edges from $E_F(O_i)$,
note that there must be at least one other optimal
quad that contains at most three edges from $E_F(O_i)$
as $(|E_F(O_i)|-3)\bmod 4 = 3$.

For each optimal quad that contains one or two
edges from $E_F(O_i)$, define these one or two edges to be
\emph{special} edges.
For each optimal quad that contains three
edges from $E_F(O_i)$, select one of these three edges
arbitrarily and define it to be a \emph{special} edge.
There are at least two special edges
in each $E_F(O_i)$, $1\le i\le d$, and hence at
least $2d$ special edges in total. More precisely, we refer to these special edges as the edge-set SE, and partition it into two subsets: those special edges occurring in a quad with cost 4 (the set SE4), and those special edges occurring in a quad with cost at least 5 (the set SE5). Clearly:
\begin{equation}
\label{ineq:2d}
2d \leq |SE4|+ |SE5|\text{.}
\end{equation}

Consider a quad with cost 4 from the optimum solution. It consists of four edges of $F$. Since $F$ is a connected simple graph there are only two possible subgraphs induced by $Q$, as depicted in Figure~\ref{fig:Q4}. These four edges can be in the sets $E_F(O_i)$ for some $1\le i\le d$, the set $E_F(X)$, and the sets $E_F(C)$ for even components $C$ of $H-X$.
We now define types of quads of cost 4 depending on how many edges are in which set.

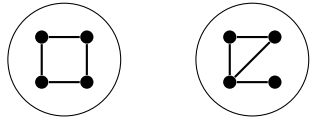
\begin{figure}[ht]
\centering
\begin{tikzpicture}

\def\xShi{0};
\def\yShi{0};
\def\inter{0.3};
\def\size{0.75};
\draw (\xShi,\yShi) circle (\size cm);
\node(q11) at (\xShi -\inter,\yShi+\inter) [circle,fill=black,inner sep=0pt,minimum size=5pt] {};
\node(q12) at (\xShi +\inter,\yShi+\inter) [circle,fill=black,inner sep=0pt,minimum size=5pt] {};
\node(q13) at (\xShi -\inter,\yShi-\inter) [circle,fill=black,inner sep=0pt,minimum size=5pt] {};
\node(q14) at (\xShi +\inter,\yShi-\inter) [circle,fill=black,inner sep=0pt,minimum size=5pt] {};
\draw[thick](q11)--(q13);
\draw[thick](q11)--(q12);
\draw[thick](q14)--(q13);
\draw[thick](q12)--(q14);

\def\xShi{2.5};
\def\yShi{0};
\draw (\xShi,\yShi) circle (\size cm);
\node(q31) at (\xShi -\inter,\yShi+\inter) [circle,fill=black,inner sep=0pt,minimum size=5pt] {};
\node(q32) at (\xShi +\inter,\yShi+\inter) [circle,fill=black,inner sep=0pt,minimum size=5pt] {};
\node(q33) at (\xShi -\inter,\yShi-\inter) [circle,fill=black,inner sep=0pt,minimum size=5pt] {};
\node(q34) at (\xShi +\inter,\yShi-\inter) [circle,fill=black,inner sep=0pt,minimum size=5pt] {};
\draw[thick](q31)--(q32);
\draw[thick](q31)--(q33);
\draw[thick](q32)--(q33);
\draw[thick](q33)--(q34);

\end{tikzpicture}
\caption{Quads with $\cost(Q)=4$}
\label{fig:Q4}
\end{figure}

Note that an edge from $E_F(O_i)$ cannot be incident to the same
vertex as an edge from $E_F(O_j)$ for $j\neq i$
because otherwise $H$ would contain an edge
between $O_i$ and $O_j$. Similarly, an edge
from $E_F(O_i)$ cannot be incident to the same
vertex as an edge from $E_F(C)$ where $C$ is
an even component of $H-X$. The only edges
that can share endpoints with edges in
$E_F(O_i)$ are those in $E_F(X)$.

We tabulate the different types of quads with cost 4 in Table~\ref{TaQuadscosts}.
Thus, a quad with cost 4 with a special edge must be of type 1, 2, 3, 4 or 5. For each of these types, the number of edges from $E_F(X)$ is at least the number of special edges in the quad. Thus,
\begin{equation}
\label{ineq:EX}
|E_F(X)| \geq |SE4|\text{.}
\end{equation}

\begin{table}[ht]
	
	\centering
	\begin{tabular}{|l||l|l|l||l|l|}
		\hline
		Type of & Number of edges & & & Cost & Number of\\
		quad & in $E_F(O_i)$ &in $E_F(X)$ & in $E_F(C)$ & & special edges \\
		\hline
		\  &\  &\ &\ &\ &\ \\[-11pt]
		1 &  3 & 1 &  & 4 & 1 \\[2pt]
		2 &  2 & 2 &  & 4 & 2\\[2pt]
		3 &  1, 1 & 2 &  & 4 & 2\\[2pt]
		4 &  1 & 2 & 1 & 4 & 1\\[2pt]
		5 &  1 & 3 &  & 4 & 1 \\[2pt]
		\hline
	\end{tabular}
	\caption{Overview of different types of quads with cost 4, containing at least 1 edge from $E_F(O_i)$. The entry ``1,1'' for quad type 3 means that there is one edge from $E_F(O_i)$ and one edge from $E_F(O_{i'})$ for $i\neq i'$}
	\label{TaQuadscosts}
\end{table}

Further, since $|E_F(X)|=2|X|$, it follows from (\ref{ineq:EX}) and (\ref{ineq:2d}) that
$
|SE5| \geq 2d-2|X|
$.
Thus, the number of quads of cost at least~$5$ is at least $\frac{2d-2|X|}{4}=\frac12(d-|X|)=k-\mu$.
\end{proof}

\section{Bad instances}
\label{sec:lowerb}
In this section, we give the instances that provide the lower bound results for
problem PQ and its special cases, as announced in Table~\ref{table:results}:
Section~\ref{sec:lowerbinstanceofPQ} presents the instance
of problem PQ$(\#1\in\{1,2\})$,
Section~\ref{sec:lowerbinstanceofPQ12} the instance
of problem PQ$(\#1=2)$, and
Section~\ref{sec:lowerbinstanceofspecialcase} the instance
of problem PQ$(\#1=2, \mbox{distinct, connected})$.
Furthermore, we illustrate in Section~\ref{sec:badgreedy} that a natural greedy algorithm (that can be seen as an alternative for algorithm $A$) has a worst-case ratio that is worse than the worst-case ratio of algorithm $A$.

\subsection{An instance of PQ\texorpdfstring{$(\#1\in \{1,2\})$}{(\#1\ in\ \{1,2\})}}
\label{sec:lowerbinstanceofPQ}%
Consider the instance $I$ consisting of the following 8 vectors, $v_1, \ldots, v_8$:
$$
\left(\begin{array}{c} 1 \\ 0 \\ 0 \\ 0\end{array}\right),
\left(\begin{array}{c} 0 \\ 1 \\ 0 \\ 0\end{array}\right),
\left(\begin{array}{c} 0 \\ 0 \\ 1 \\ 0\end{array}\right),
\left(\begin{array}{c} 0 \\ 0 \\ 0 \\ 1\end{array}\right),
\left(\begin{array}{c} 1 \\ 1 \\ 0 \\ 0\end{array}\right),
\left(\begin{array}{c} 1 \\ 1 \\ 0 \\ 0\end{array}\right),
\left(\begin{array}{c} 0 \\ 0 \\ 1 \\ 1\end{array}\right),
\left(\begin{array}{c} 0 \\ 0 \\ 1 \\ 1\end{array}\right).
$$
Since each vector contains either one or two 1's, this is an instance of PQ$(\#1\in \{1,2\})$. Clearly, the optimum solution consists of the quads $\{v_1, v_2, v_5, v_6\}$ and $\{v_3, v_4, v_7, v_8\}$, with a total cost of $\OPT(I)=4$; algorithm $A$ however, may find, as an optimum matching in the first phase, the pairs $\{v_1,v_3\}$, $\{v_2,v_4\}$, $\{v_5,v_6\}$ and $\{v_7,v_8\}$, leading to a final solution with cost $A(I)=6$. Thus, we arrive at the following observation.

\begin{obs}
\label{wcinstPQatmosttwoones}
For the instance depicted above, $cost (A) = \frac32 OPT$.
\end{obs}
Theorems~\ref{theo:PQ} and \ref{theo:PQatmosttwoones} now follow from Lemma~\ref{ubPQ} and Observation~\ref{wcinstPQatmosttwoones}.

As a remark, if we would allow all-zero vectors in the
input (which we do not allow in PQ$(\#1\in \{1,2\})$),
we can get an even smaller example with ratio $\frac32$. Let
$I$ consist of the following 8 vectors:
$$
\left(\begin{array}{c} 1 \\ 0 \end{array}\right),
\left(\begin{array}{c} 1 \\ 0 \end{array}\right),
\left(\begin{array}{c} 1 \\ 0 \end{array}\right),
\left(\begin{array}{c} 0 \\ 0 \end{array}\right),
\left(\begin{array}{c} 0 \\ 1 \end{array}\right),
\left(\begin{array}{c} 0 \\ 1 \end{array}\right),
\left(\begin{array}{c} 0 \\ 1 \end{array}\right),
\left(\begin{array}{c} 0 \\ 0 \end{array}\right).
$$
The optimal cost is $2$; algorithm $A$ however, may match the
two all-zero vectors in the first phase and get a solution
of cost~$3$.

\subsection{An instance of PQ\texorpdfstring{$(\#1=2)$}{(\#1=2)}}
\label{sec:lowerbinstanceofPQ12}%
Consider the instance $I$ consisting of the following 8 vectors, $v_1, \dots, v_8$.
$$
\left(\begin{array}{c} 1 \\ 1 \\0\\0\\0 \end{array}\right),
\left(\begin{array}{c} 1 \\ 1 \\0\\0\\0 \end{array}\right),
\left(\begin{array}{c} 1 \\ 0 \\1\\0\\0 \end{array}\right),
\left(\begin{array}{c} 0 \\ 1 \\1\\0\\0 \end{array}\right),
\left(\begin{array}{c} 0 \\ 0 \\1\\1\\0 \end{array}\right),
\left(\begin{array}{c} 0 \\ 0 \\1\\0\\1 \end{array}\right),
\left(\begin{array}{c} 0 \\ 0 \\0\\1\\1 \end{array}\right),
\left(\begin{array}{c} 0 \\ 0 \\0\\1\\1 \end{array}\right),
$$
Since each vector contains two 1's, this is an instance of PQ$(\#1=2\})$. This instance can be represented by the graph shown in Figure~\ref{fig:PQ2LB}.
\begin{figure}[ht]
\centering
\begin{tikzpicture}
\tikzstyle{mynode}=[circle,draw, fill = white]
\node (n1) at (0,4)  {1};
\node (n11) at (0,3.9){};
\node (n12) at (0,4.1){};
\node (n2) at (2,4)  {2};
\node (n21) at (2,3.9){};
\node (n22) at (2,4.1){};
\node (n3) at (1,2) [mynode]  {3};
\node (n4) at (0,0){4};
\node (n41) at (0,-0.1){};
\node (n42) at (0,0.1){};
\node (n5) at (2,0) {5};
\node (n51) at (2,-0.1){};
\node (n52) at (2,0.1){};
\draw (n11) -- (n21);
\draw (n12) -- (n22);
\draw (n1) -- (n3);
\draw (n2) -- (n3);
\draw (n3) -- (n4);
\draw (n3) -- (n5);
\draw (n41) -- (n51);
\draw (n42) -- (n52);
\node at (0,4) [mynode] {1};
\node  at (2,4) [mynode] {2};
\node  at (0,0) [mynode] {4};
\node at (2,0) [mynode] {5};
\end{tikzpicture}
\caption{An instance of PQ$(\#1=2)$}
\label{fig:PQ2LB}
\end{figure}
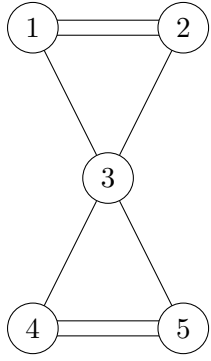

The optimal solution for this instance has cost $6$,
with the two quads
$$\{v_1,v_2,v_3,v_4\} = \{ (1,2), (1,2), (1,3), (2,3) \},$$
$$\{v_5,v_6,v_7,v_8\} = \{ (3,4), (3,5), (4,5), (4,5) \}.$$

The algorithm may in the first phase construct a matching
with cost $10$ consisting of the following pairs:
\begin{eqnarray*}
\{v_1,v_2\}=\{ (1,2),(1,2) \}, \mbox{ } \{v_3,v_5\}=\{ (1,3),(3,4) \}, \\
\{v_4,v_6\}=\{ (2,3),(3,5) \}, \mbox{ } \{v_7,v_8\}=\{ (4,5),(4,5) \}.
\end{eqnarray*}
Any two pairs share at most $1$ node. Hence, the total savings
that can be made in the second matching are at most $2$, so
by Corollary~\ref{cor:sav} we have $A(I)\ge 8$. Hence,
the worst-case approximation ratio of $A$ is at least $8/6=4/3$.

\begin{obs}
\label{wcinstPQexactlytwoones}
For the instance depicted in Figure~\ref{fig:PQ2LB}, $cost (A) = \frac43 OPT$.
\end{obs}

Theorem~\ref{theo:PQexactlytwoones} now follows from Lemma~\ref{ubPQexactlytwoones} and Observation~\ref{wcinstPQexactlytwoones}.

\subsection{An instance of PQ\texorpdfstring{$(\#1=2, \mbox{distinct, connected})$}{(\#1=2,distinct,connected)}}
\label{sec:lowerbinstanceofspecialcase}%
Consider the instance $I$ consisting of the following 8 vectors $v_1, \dots, v_8$.
$$
\left(\begin{array}{c} 1 \\ 1 \\0\\0\\0\\0\\0 \end{array}\right),
\left(\begin{array}{c} 1 \\ 0 \\1\\0\\0\\0\\0 \end{array}\right),
\left(\begin{array}{c} 0 \\ 1 \\0\\1\\0\\0\\0 \end{array}\right),
\left(\begin{array}{c} 0 \\ 0 \\1\\1\\0\\0\\0 \end{array}\right),
\left(\begin{array}{c} 0 \\ 0 \\0\\1\\1\\0\\0 \end{array}\right),
\left(\begin{array}{c} 0 \\ 0 \\0\\1\\0\\1\\0 \end{array}\right),
\left(\begin{array}{c} 0 \\ 0 \\0\\0\\1\\0\\1 \end{array}\right),
\left(\begin{array}{c} 0 \\ 0 \\0\\0\\0\\1\\1 \end{array}\right).
$$
Since each vector contains two 1's, the vectors are pairwise distinct, and the induced graph is connected, this is an instance of PQ$(\#1=2\}, \mbox{distinct, connected})$. The instance can be represented by the graph shown in Figure~\ref{fig:PQGconLB}.

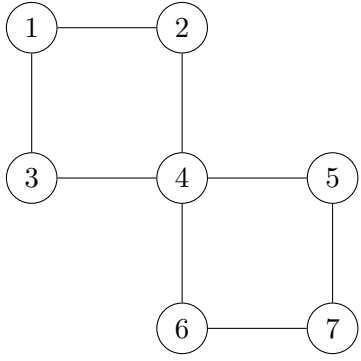
\begin{figure}[ht]
\centering
\begin{tikzpicture}
\tikzstyle{mynode}=[circle,draw]
\node (n1) at (0,4) [mynode] {1};
\node (n2) at (2,4) [mynode] {2};
\node (n3) at (0,2) [mynode] {3};
\node (n4) at (2,2) [mynode] {4};
\node (n5) at (4,2) [mynode] {5};
\node (n6) at (2,0) [mynode] {6};
\node (n7) at (4,0) [mynode] {7};
\draw (n1) -- (n2);
\draw (n1) -- (n3);
\draw (n2) -- (n4);
\draw (n3) -- (n4);
\draw (n4) -- (n5);
\draw (n4) -- (n6);
\draw (n5) -- (n7);
\draw (n6) -- (n7);
\end{tikzpicture}
\caption{An instance of PQ$(\#1=2, \mbox{distinct, connected})$}
\label{fig:PQGconLB}
\end{figure}

The optimal solution for this instance has cost $8$,
with the two quads
$$\{v_1,v_2,v_3,v_4\} = \{ (1,2), (1,3), (2,4), (3,4) \},$$
$$\{v_5,v_6,v_7,v_8\} = \{ (4,5), (4,6), (5,7), (6,7) \}.$$

Algorithm $A$ may, in the first phase, construct a matching
with cost $12$ consisting of the following pairs:
\begin{eqnarray*}
\{v_1,v_2\}=\{ (1,2),(1,3) \},
\{v_3,v_5\}=\{ (2,4),(4,5) \}, \\
\{v_4,v_6\}=\{ (3,4),(4,6) \},
\{v_7,v_8\}=\{ (5,7),(6,7) \}.
\end{eqnarray*}
Any two pairs share at most $1$ node. Hence, the total savings
that can be made in the second matching are at most $2$, so
by Corollary~\ref{cor:sav} we have $A(I)\ge 10$. Hence,
the worst-case ratio of $A$ is at least $10/8=5/4$.

\begin{obs}
\label{wcinstPQexactlytwoonesanddistandconn}
For the instance depicted in Figure~\ref{fig:PQGconLB}, $cost (A) = \frac54 OPT$.
\end{obs}
Theorem~\ref{theo:PQexactlytwoonesanddisandcon} now follows from Lemma~\ref{ubPQexactlytwoonesanddistandconn} and Observation~\ref{wcinstPQexactlytwoonesanddistandconn}.

\subsection{Bad instances for a natural greedy algorithm}
\label{sec:badgreedy}%
In this section, we show that the worst-case ratio of a natural greedy algorithm is worse than the worst-case ratio of algorithm $A$.

An informal description of the greedy algorithm for problem PQ (and its special cases) is as follows: repeatedly select, among all possible quads, a quad with lowest cost, and remove the vectors in the selected quad from the instance; stop when no more vectors remain.

Below we present instances of problem PQ, as well as of its special case PQ$(\#1=2, \mbox{distinct},$ con\-nect\-ed$)$, showing that the worst-case performance of this greedy algorithm is worse than the worst-case performance of algorithm $A$.

\paragraph*{An instance of PQ}

Consider the following instance $I$ of PQ consisting of the following 12 vectors, $v_1, \ldots, v_{12}$:
$$
\left(\begin{array}{c} 0 \\ 0 \\ 0 \end{array}\right),
\left(\begin{array}{c} 0 \\ 0 \\ 0 \end{array}\right),
\left(\begin{array}{c} 0 \\ 0 \\ 0 \end{array}\right),
\left(\begin{array}{c} 1 \\ 0 \\ 0 \end{array}\right),
\left(\begin{array}{c} 1 \\ 0 \\ 0 \end{array}\right),
\left(\begin{array}{c} 1 \\ 0 \\ 0 \end{array}\right),
$$
$$
\left(\begin{array}{c} 0 \\ 1 \\ 0 \end{array}\right),
\left(\begin{array}{c} 0 \\ 1 \\ 0 \end{array}\right),
\left(\begin{array}{c} 0 \\ 1 \\ 0 \end{array}\right),
\left(\begin{array}{c} 0 \\ 0 \\ 1 \end{array}\right),
\left(\begin{array}{c} 0 \\ 0 \\ 1 \end{array}\right),
\left(\begin{array}{c} 0 \\ 0 \\ 1 \end{array}\right).
$$

Clearly, an optimum solution consists of the quads $\{v_1, v_4, v_5, v_6\}$, $\{v_2, v_7, v_8, v_9\}$ and $\{v_3$, $v_{10}$, $v_{11}$, $v_{12}\}$, with a total cost of $\OPT(I)=3$.
The greedy algorithm however, may first select quad $\{v_1,v_2,v_3,v_4 \}$ with cost $1$. Then the cheapest possible quad is one of cost 2 and the greedy algorithm may select $\{v_7,v_8, v_{11},v_{12} \}$. The remaining vectors form a quad of cost 3. Thus, the greedy algorithm finds a solution of cost $6 = 2 \cdot OPT(I)$.

\paragraph*{An instance of PQ$(\#1=2, \mbox{distinct, connected})$}

Consider the following instance $I$ of PQ$(\#1=2, \mbox{distinct, connected})$ consisting of 8 vectors 
represented in a graph shown in Figure~\ref{fig:PQGconLB2} (recall that a vector in PQ$(\#1=2, \mbox{distinct,} \\ \mbox{connected})$ corresponds to an edge in a simple graph).

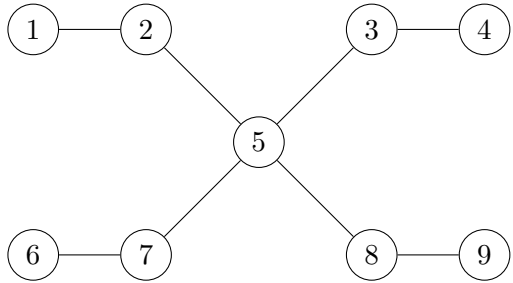
\begin{figure}[ht]
\centering
\begin{tikzpicture}
\tikzstyle{mynode}=[circle,draw]
\node (n1) at (0,3) [mynode] {1};
\node (n2) at (1.5,3) [mynode] {2};
\node (n3) at (4.5,3) [mynode] {3};
\node (n4) at (6,3) [mynode] {4};
\node (n5) at (3,1.5) [mynode] {5};
\node (n6) at (0,0) [mynode] {6};
\node (n7) at (1.5,0) [mynode] {7};
\node (n8) at (4.5,0) [mynode] {8};
\node (n9) at (6,0) [mynode] {9};
\draw (n1) -- (n2);
\draw (n2) -- (n5);
\draw (n3) -- (n5);
\draw (n3) -- (n4);
\draw (n6) -- (n7);
\draw (n7) -- (n5);
\draw (n8) -- (n5);
\draw (n8) -- (n9);
\end{tikzpicture}
\caption{An instance of PQ$(\#1=2, \mbox{distinct, connected})$}
\label{fig:PQGconLB2}
\end{figure}

An optimal solution for this instance has cost $10$, with the two quads $\{ (1,2), (2,5), (3,5), \\ (3,4) \}$ and $\{ (6,7), (5,7), (5,8), (8,9) \}$, each having cost 5.

Since the instance features no quad with cost 4, the greedy algorithm may first select the following quad
with cost $5$: $\{ (2,5), (3,5), (5,7), (5,8) \}$. Next, what remains is a quad of cost $8$:
$\{ (1,2), (3,4), (6,7), (8,9) \}$.

Hence, the worst-case ratio of the greedy algorithm is at least $13/10$, which is larger than the $5/4$ approximation guarantee for algorithm $A$.

\section{Conclusion}
\label{sec:conclusion}
We have studied the worst-case behavior of a natural algorithm for partitioning a given set of vectors into quadruples. Informally, by running a matching algorithm once, we find pairs, and by running it one more time, we match the pairs into quadruples. Under the specific cost-structure studied here, we have shown the precise worst-case behavior of this method for all cases except PQ$(\#1=2,\mbox{distinct})$, where a small gap remains.

It is a natural question to study an extension where we form clusters consisting of $2^s$ vectors for some given integer $s \geq 2$. Indeed, if we form groups of size $2^s$ by running $s$ rounds of matching, the worst-case ratio is easily seen to be bounded by $2^{s-1}$. To explain this, let $M$ be the minimum-cost matching of the first round. Then $A(I)\le \cost(M)$ and $\OPT(I)\ge \cost(M)/2^{s-1}$ as the cost of the optimum (viewed as being constructed in $s$ rounds) is at least $\cost(M)$ after the first round and could then halve in each further round. Moreover, since we have shown that the cost of the algorithm after two rounds is at most $\frac32$ times the optimal cost after two rounds, we get a ratio of $\frac32 \times 2^{s-2} = 3 \times 2^{s-3}$. We leave the question of finding the worst-case ratio for arbitrary $s$ as an open problem.

\bibliography{grouping}
\end{document}